\documentclass[prodmode]{acmsmall} % Aptara syntax
\usepackage{hyperref}
\usepackage[english]{babel}
\usepackage[matrix,arrow]{xy}

\usepackage[utf8]{inputenc}
\usepackage{stmaryrd}
	       {\begin{array}{@{\hspace{5em}}p{5em}p{50em}}}{\end{array}}

\newcommand{\upA}{A{\uparrow}}
\newcommand{\ExpTime}{\textsc{ExpTime}\xspace}
\usepackage[layout=footnote,draft]{fixme}

\FXRegisterAuthor{yv}{ayv}{YV}	% Yde
\FXRegisterAuthor{ls}{als}{LS}	% Lutz

\newcommand{\Gm}{\mathsf{G}}

\newcommand{\hearts}{\heartsuit}
\newcommand{\Set}{\mathsf{Set}}

\newcommand{\lsem}{\llbracket}
\newcommand{\rsem}{\rrbracket}
\newcommand{\Pow}{\mathcal{P}}

\newcommand{\Land}{\bigwedge}
\newcommand{\Lor}{\bigvee}
\newcommand{\entails}{\vdash}

\usepackage{amsmath}
\usepackage{wasysym}

\usepackage{amssymb}
\usepackage{xspace}

\newcommand{\FA}{\mathfrak{A}}

\newcommand{\CondArrow}{\Rightarrow}

\newcommand{\CF}{\mathit{Cf}}

\newcounter{blubber}

\newenvironment{myenumerate}
{\begin{enumerate}
\setlength{\itemsep}{0pt}
    \setlength{\leftmargin}{0pt}
    \setlength{\itemindent}{0pt}
}
{\end{enumerate}}

\newcommand{\PLentails}{\entails_{\mi{PL}}}

\newcommand{\mi}[1]{\mathit{#1}}

\newcommand{\Sem}[1]{{[\![#1]\!]}}

\newlength{\croutw}
\newlength{\crouth}
\newcommand{\crossout}[1]%
        {\settowidth{\croutw}{$#1$}\settoheight{\crouth}{$#1$}#1%
        \hspace{-1.0\croutw}\raisebox{0.3\crouth}{\rule{\croutw}{0.1ex}}}

\newcommand{\commentout}[1]{\ignorespaces}

\newcommand{\infrule}[2]{\frac{#1}{#2}}

\newcommand{\bit}{\begin{itemize}}
\newcommand{\eit}{\end{itemize}}

\newenvironment{rmenumerate}%
        {\begin{enumerate}}%
        {\end{enumerate}}
        {\begin{enumerate}}%
        {\end{enumerate}}
        {\begin{enumerate}}%
        {\end{enumerate}}
        {\begin{enumerate}}%
        {\end{enumerate}}
\newcommand{\rightqed}{\hfill\mbox{}\qed}

\newcommand{\Cls}{\mathcal}

\newcommand{\Op}{{op}}

\newcommand{\lrule}[3]{(#1)\;\;\infrule{#2}{#3}}

\newcommand{\CM}{{\Cls M}}
\newcommand{\CO}{{\Cls O}}

\newcommand{\CT}{{\Cls T}}

\newcommand{\mystrut}[1]{\rule[#1]{0cm}{0.1cm}}
\newcommand{\dualBox}{\dual{\Box\mystrut{3pt}}}

\newcommand{\id}{{id}}

\newcommand{\powerset}{{\mathcal P}}

\newcommand{\modimpl}{\to}
\newcommand{\modiff}{\leftrightarrow}

        {\smallskip\par\noindent\qquad$\begin{array}{@{}l@{{}={}}l}}%
        {\end{array}$\par\noindent}

\newlength{\myboxwidth}
	{\end{center}\end{figure}}

\newcommand{\Lang}{\mathcal{L}}	
\newcommand{\FLang}{\mathcal{F}}

\newcommand{\plbox}[1]{[#1]}

\newcommand{\gldiamond}[1]{\Diamond_{#1}}
\newcommand{\glbox}[1]{\square_{#1}}

\newcommand{\Nat}{{\mathbb{N}}}

\newcommand{\Rat}{{\mathbb{Q}}}

\newcommand{\PDist}{\mathcal{D}}

\newcommand{\Bag}{\mathcal{B}}

\newcommand{\Prop}{\mathsf{Prop}}

\newcommand{\Rules}{\mathcal{R}}

\newcommand{\contrapow}{\mathcal{Q}}

\newcommand{\EXP}{\textsc{ExpTime}\xspace}

\newcommand{\CK}{\mathit{CK}}

\newcommand{\Space}{\mathcal{S}}

\newcommand{\ATLdiamond}[1]{{\langle\!\langle#1\rangle\!\rangle}}

\newcommand{\until}{U}

\newcommand{\dual}[1]{\overline{{#1}}}
\newcommand{\CA}{\mathcal{A}}

\newcommand{\FL}{\mathit{FL}}
\newcommand{\eat}[1]{}

\begin{document}

\markboth{L.\ Schröder and Y.\ Venema}{Completeness of Flat
  Coalgebraic Fixpoint Logics} \title{Completeness of Flat Coalgebraic Fixpoint
  Logics} \author{LUTZ SCHRÖDER \affil{Friedrich-Alexander-Universität
    Erlangen-Nürnberg} YDE VENEMA \affil{ILLC, Universiteit van
    Amsterdam} }

\begin{abstract}
\noindent Modal fixpoint logics traditionally play a central role in
computer science, in particular in artificial intelligence and
concurrency. The $\mu$-calculus and its relatives are among the
most expressive logics of this type. However, popular fixpoint
logics tend to trade expressivity for simplicity and readability,
and in fact often live within the single variable fragment of the
$\mu$-calculus. The family of such \emph{flat} fixpoint logics
includes, e.g., LTL, CTL, and the logic of common
knowledge. Extending this notion to the generic semantic framework
of \emph{coalgebraic logic} enables covering a wide range of logics
beyond the standard $\mu$-calculus including, e.g., flat fragments
of the graded $\mu$-calculus and the alternating-time $\mu$-calculus
(such as alternating-time temporal logic ATL), as well as
probabilistic and monotone fixpoint logics. We give a generic proof
of completeness of the Kozen-Park axiomatization for such flat
coalgebraic fixpoint logics.
\end{abstract}




\maketitle             

\section{Introduction}

\noindent Many of the most well-known logics in program verification,
concurrency, and other areas of computer science and artificial
intelligence can be cast as modal fixpoint logics, that is, embedded
into some variant of the $\mu$-calculus. Typical examples are
PDL~\cite{Pratt76} where, say, the formula $\langle a^ *\rangle p$
(`$p$ can be reached by finite iteration of $a$') can be expressed as
the least fixpoint
\begin{equation*}
\mu X.\, p\lor\langle a\rangle
X;
\end{equation*}
CTL~\cite{EmersonClarke82}, whose formula $AFp$ (`$p$ eventually holds
on all paths') is just the fixpoint
\begin{equation*}
  \mu X.\,p\lor\Box X;
\end{equation*}
and the common knowledge operator $C$ of epistemic
logic~\cite{Lewis69}, where $Cp$ (`it is common knowledge that $p$')
can be expressed as the fixpoint
\begin{equation*}
\nu X.\,\Land_{i=1}^nK_i(p\land X)
\end{equation*}
with $n$ the number of agents and $K_i$ read as `agent $i$ knows
that'. A common feature of these examples is that they trade off
expressivity for simplicity of expression in comparison to the full
$\mu$-calculus.

One of the reasons why the full $\mu$-calculus is hard for both end
users and logicians 
%(although of the same worst-case
%complexity \EXP as typical less expressive fixpoint logics) 
is that it requires keeping track of bound variables. Indeed we note that
the simpler logics listed above (in the case of PDL, the
$*$-nesting-free fragment) live in the single-variable fragment of the
$\mu$-calculus (a subfragment of the alternation-free
fragment~\cite{EmersonLei86}), which is precisely what enables one to
abandon variables altogether in favour of variable-free fixpoint
operators such as $AF$ or $C$. We refer to logics that embed
into a single-variable $\mu$-calculus as \emph{flat fixpoint
  logics}~\cite{SantocanaleVenema10}.

Here, we study flat fixpoint logics in the more general setting of
\emph{coalgebraic logic}. Coalgebra has emerged as the right level of
generality for a unified treatment of a wide range of modalities with
seemingly disparate semantics beyond the realm of pure relational
structures. Examples include monotone modalities~\cite{Chellas80},
probabilistic modalities~\cite{LarsenSkou91}, graded
modalities~\cite{Fine72,DAgostinoVisser02},
coalitional/alternating-time modalities~\cite{AlurEA02,Pauly02}, and
various non-monotonic
conditionals~\cite{FriedmanHalpern94,OlivettiEA07}. The semantics of
coalgebraic logic is parametrized over the choice of an endofunctor on
the category of sets, whose coalgebras play the role of
frames. Besides standard Kripke frames, the notion of coalgebra
encompasses, e.g., Markov chains, weighted automata, multigraphs,
neighbourhood frames, selection function frames~\cite{Chellas80}, and
concurrent game structures~\cite{AlurEA02}. Generic completeness
results in coalgebraic logic are parametrized over sets of rules or
axioms that satisfy a local form of completeness called \emph{one-step
  completeness}. That is, they require completeness of a restricted
logic without fixpoints and nesting of modalities that is interpreted
over mere elements of the functor rather than full-blown frames or
models (e.g.\ in the relational base case, over a single subset of the
base set thought of as a local view on a set of successors), a
condition that is typically quite easy to establish. In fact, suitable
one-step complete axiomatizations for many examples can already be
found in the literature on coalgebraic
logic~\cite{Pattinson03,Schroder07,SchroderPattinson09a,KupkePattinson10}.

In our \emph{flat coalgebraic fixpoint logics} one thus can express
statements such as `the coalition $C$ of agents can maintain $p$
forever', `the present state is the root of a binary tree all whose
leaves satisfy $p$', or `$p$ is commonly believed with reasonable
certainty'. In particular, we cover flat fragments of the graded
$\mu$-calculus~\cite{KupfermanEA02} and the alternating-time
$\mu$-calculus (AMC)~\cite{AlurEA02}; one such flat fragment is
alternating-time temporal logic (ATL). 

Our main result on flat coalgebraic fixpoint logics is completeness of
the natural axiomatization that makes the fixpoint definitions
explicit, generalizing the well-known Kozen-Park axiomatization. The
axiomatization is parametric both w.r.t.\ the coalgebraic branching
type and the choice of flat fragment, under mild restrictions on the
form of fixpoint operators. This result generalizes results by
Santocanale and Venema~\citeyear{SantocanaleVenema10} to the level of
coalgebraic logic, and relies on the notion of
\emph{$\CO$-adjointness}~\cite{Santocanale08} to prove that fixpoints
in the Lindenbaum algebra are \emph{constructive}, i.e.\ approximable
in $\omega$ steps. The crucial ingredient here are \emph{one-step
  cutfree complete} rule
sets~\cite{SchroderPattinson09a,PattinsonSchroder10}. These enable
generalizations of both the key \emph{rigidity lemma} and the
$\CO$-adjointness theorem of~\cite{SantocanaleVenema10}, the latter to
the effect that \emph{all uniform-depth modal operators are
  $\CO$-adjoint}.

Our completeness result follows a long tradition of non-trivial
completeness proofs for fixpoint logics, e.g.\
PDL~\cite{KozenParikh81,Segerberg82}, CTL~\cite{EmersonHalpern85},
LTL~\cite{GabbayEA80,LichtensteinPnueli00}, the aconjunctive
$\mu$-calculus~\cite{Kozen83}, and the full
$\mu$-calculus~\cite{Walukiewicz00}. Note that all these results are
independent, as completeness is not in general inherited by sublogics,
and in fact employ quite different methods.  Instantiating our generic
results to concrete logics yields new results in nearly all cases that
go beyond the classical relational $\mu$-calculus, noting that neither
Kupferman et al.~\citeyear{KupfermanEA02} nor Cirstea et
al.~\citeyear{CirsteaEA11} cover axiomatizations. In particular, we
obtain for the first time a completeness result for graded fixpoint
logics, i.e.\ fragments of the graded $\mu$-calculus, and we
generalize the completeness of ATL~\cite{GorankoVanDrimmelen06} to
arbitrary flat fragments of the AMC.

\paragraph*{Further Related Work}

The present paper is an extended and revised version of a previous
conference paper~\cite{SchroderVenema10}. The technical treatment
differs from the conference version in that we weaken the assumptions
of the generic completeness theorem to require only a one-step
complete rule set instead of a one-step cutfree complete one;
moreover, we opt for a propositional basis with unrestricted negation,
converting to negation normal form only for purposes of the model
construction in Section~\ref{sec:models}. Flat coalgebraic fixpoint
logics are fragments of coalgebraic $\mu$-calculi, and as such known
to be decidable in \EXP under reasonable
assumptions~\cite{CirsteaEA11}. A tableau-based global caching
algorithm for flat fixpoint logics has recently been developed by
Hausmann and one of the authors
(Schröder)~\citeyear{HausmannSchroder15}, so we omit discussion of the
(less practical) tableau algorithm given in the conference version.

\paragraph*{Organization} We recall the requisite background in
coalgebraic logic and introduce the syntax of semantics of flat
coalgebraic fixpoint logics in Section~\ref{sec:prelim}. We proceed to
discuss the generic axiomatization in Section~\ref{sec:ax}. In
Section~\ref{sec:constructive}, we prove the central $\CO$-adjointness
theorem, and then present the ensuing model construction in
Section~\ref{sec:models}. Section~\ref{sec:conclusions} concludes.

\section{Flat Coalgebraic Fixpoint Logics}
\label{sec:prelim}
\noindent We briefly recall the generic framework of coalgebraic modal
logic~\cite{Pattinson04,Schroder07}, and define its extension with
flat fixpoint operators, a fragment of the coalgebraic
$\mu$-calculus~\cite{CirsteaEA11}. We present the calculus in a form
that includes negation, and therefore need to pay attention to
positive and negative occurrences of variables and subformulas; to
avoid excessive repetition, we fix these notions for all notions of
formula that include negation and possibly a notion of (propositional)
variable:
\begin{definition}
  In any logic with negation $\neg$, an occurrence of a subformula in
  a formula is \emph{positive} if it is in the scope of an even number
  of negations $\neg$, and otherwise \emph{negative}. In logics
  featuring a notion of (propositional) \emph{variable}, a formula
  $\phi$ is \emph{positive} (\emph{negative}) in a variable $x$ if all
  occurrences of $x$ in $\phi$ are positive (negative).
\end{definition}
\noindent For the rest of the paper, we fix a countably infinite set
$V$ of \emph{variables}, with a single distinguished \emph{recursion
  variable} that is always called $x$; all other variables are
\emph{parameter variables}, typically called $p$ or
$p_1,p_2,\dots$. The first parameter of the syntax of a \emph{flat
  coalgebraic fixpoint logic} is a \emph{(modal) similarity type}
$\Lambda$, i.e.\ an at most countable set of modal operators with
associated finite arities. The set of \emph{modal fixpoint schemes}
$\gamma,\delta$ is given by the grammar
\begin{equation*}
  \gamma,\delta::=\bot\mid v\mid\neg\gamma\mid \gamma\land\delta\mid \hearts(\gamma_1,\dots,\gamma_n)
\end{equation*}
where $\hearts\in\Lambda$ is $n$-ary and $v\in V$; we require
additionally that fixpoint schemes are positive in all variables (this
is essential in case of the recursion variable $x$ to ensure existence
of fixpoints; for parameter variables, it is a mere technical
convenience, as negative occurrences of a parameter variable can be
replaced with positive occurrences of a fresh variable, with negation
moved into the parameter formula). Further Boolean operations $\top$,
$\lor$, $\modimpl$, $\modiff$ are defined as usual. Moreover, we
abbreviate
\begin{equation*}
  \dual\hearts(\phi_1,\dots,\phi_n)=\neg\hearts(\neg\phi_1,\dots,\neg\phi_n)
\end{equation*}
for $\hearts\in\Lambda$, and refer to $\dual\hearts$ as the
\emph{dual} of $\hearts$. We intend variables as place holders for
arguments and parameters of formulas defining fixpoint operators; as
such, they serve only technical purposes and will not form part of the
actual fixpoint language defined below (other than as part of modal
fixpoint schemes indexing flat fixpoint operators). In particular,
propositional variables should not be confused with propositional
\emph{atoms}, which can appear also in actual formulas. Propositional
atoms are incorporated into the modal similarity type $\Lambda$ as
nullary operators when needed; this approach not only simplifies the
technical presentation but it also enhances generality in that it
allows covering logics that do not have propositional atoms, such as
Hennessy-Milner logic.

The second syntactic parameter is a set $\Gamma$ of modal fixpoint
schemes $\gamma$ determining the choice of fixpoint operators. We
require that all $\gamma\in\Gamma$ are \emph{guarded}, i.e.\ that all
occurrences of the recursion variable $x$ are under the scope of at
least one modal operator; this is not an essential restriction as
every $\mu$-calculus formula is provably equivalent to a guarded
formula~\cite{Walukiewicz00}. We denote substituted formulas
$\gamma[\phi_1/p_1;\dots;\phi_n/p_n;\psi/x]$ as
$\gamma(\phi_1,\dots,\phi_n,\psi)$. The set
$\FLang_\sharp(\Lambda,\Gamma)$ or just $\FLang_\sharp$ of
\emph{(fixpoint) formulas} $\phi,\psi$ is then defined by the grammar
\begin{equation*}
  \phi,\psi::=\bot\mid \neg\phi\mid\phi\land\psi\mid
  \hearts(\phi_1,\dots,\phi_n)\mid
  \sharp_\gamma(\phi_1,\dots,\phi_n)
\end{equation*}
where $\hearts\in\Lambda$ is $n$-ary and $\gamma\in\Gamma$. The
operator $\sharp_\gamma$ takes  least fixpoints
\begin{equation*}
  \sharp_\gamma(\phi_1,\dots,\phi_n)=\mu x.\gamma(\phi_1,\dots,\phi_n,x).
\end{equation*}
The name \emph{flat} for the fixpoint operators $\sharp_\gamma$
relates to the fact that modal fixpoint schemes $\gamma$ do not
contain fixpoint operators. Note however that nesting of flat fixpoint
operators is unrestricted, i.e.\ the $\phi_i$ can be arbitrary
fixpoint formulas in $\sharp_\gamma(\phi_1,\dots,\phi_n)$. We
introduce greatest fixpoint operators as duals of least fixpoint
operators: for a modal fixpoint scheme $\gamma$ we denote by
$\dual\gamma$ its \emph{dual}, i.e.~the modal fixpoint scheme
$\neg\gamma\sigma$ where $\sigma(v)=\neg v$ for all variables $v$. We
then define the greatest fixpoint operator $\flat_{\dual\gamma}$ by
\begin{equation*}
  \flat_{\dual\gamma}(\phi_1,\dots,\phi_n)= \neg\sharp_{\gamma}(\neg\phi_1,\dots,\neg\phi_n)
\end{equation*}
so that
$\flat_{\dual\gamma}(\phi_1,\dots,\phi_n)=\nu
x.\,\dual\gamma(\phi_1,\dots,\phi_n,x)$.
% E.g. flat_{p/\<>x}phi = -#_{-(p/\<>-x)}phi
% = -#_{-p\/[]x}phi = -#_{p\/[]x}-phi
% = -AF-phi = EG phi = nu x. phi /\ <>x.

A standard example of a flat fixpoint logic is CTL, whose operators
$AU,EG$ can be coded as
\begin{equation*}
  A[\phi\, \until\,\psi]=\sharp_{(p_2\lor (p_1\land\Box x))}(\phi,\psi)\quad\text{and}\quad EG\,\phi =\flat_{p\land\Diamond x}\phi.
\end{equation*}
Note here that $p\land\Diamond x$ is equivalent to the dual
$\neg(\neg p\lor\Box\neg x)$ of $p\lor\Box x$.

Syntactically,~$\sharp_\gamma$ is regarded as an atomic operator; in
particular, occurrences of variables in $\gamma$ do not count as
occurrences in formulas $\sharp_\gamma\phi$. For the sake of
readability, we restrict the further technical development to unary
modalities~$\hearts\in\Lambda$ and unary fixpoint operators, i.e.\ we
assume that every $\gamma\in\Gamma$ has only one parameter variable,
\emph{denoted by $p$ throughout}; the extension to higher arities is a
mere notational issue, and in fact we continue to use higher-arity
modalities and fixpoint operators in the examples. Note that we have
not included variables in the definition of fixpoint formulas. A
\emph{(fixpoint) formula with variables} is an expression of the form
$\gamma\sigma$, where $\gamma$ is a modal fixpoint scheme and $\sigma$
is a substitution of some of the variables in $\gamma$ with fixpoint
formulas (i.e.\ variables never appear under fixpoint operators).  In
the following, the term \emph{formula} will refer to fixpoint formulas
without variables unless variables are explicitly mentioned. We
sometimes indicate the occurrence of a variable $v$ in a formula
$\psi$ by writing $\psi(v)$, and then write $\psi(\rho)$ for the
formula obtained by substituting a formula $\rho$ for $v$. For a modal
fixpoint scheme $\gamma(p,x)$, we denote the function taking a formula
$\psi$ to $\gamma(\phi,\psi)$ by $\gamma(\phi)$, and by
$\gamma(\phi)^k$ its $k$-fold iteration.

The logic is further parametrized \emph{semantically} over the
underlying class of systems and the interpretation of the modal
operators. The former is determined by the choice of a functor
\begin{equation*}
  T:\Set\to\Set,
\end{equation*}
i.e.\ an operation $T$ that maps sets $X$ to sets $TX$ and functions
$f:X\to Y$ to functions $Tf:TX\to TY$, preserving identities and
composition.  The role of models is then played by
\emph{$T$-coalgebras}, i.e. pairs $(X, \xi)$ where $X$ is a set of
\emph{states} and
\begin{equation*}
\xi: X \to T X
\end{equation*}
is the \emph{structure map}. Thinking of $TX$ informally as a
parametrized datatype over $X$, we regard $\xi$ as associating with
each state $x$ a structured collection $\xi(x)$ of successor states
and observations. E.g.\ for $TX = \Pow{X}\times\Pow U$, with $U$ a
fixed set of propositional atoms and $\Pow$ denoting the covariant
powerset functor (with $\Pow f(A)=f[A]$ for $f:X\to Y$ and
$A\in\Pow(X)$), we obtain that $T$-coalgebras are Kripke models, as
they associate with each state a set of successor states and a set of
valid propositional atoms. Our main interest here is in examples
beyond Kripke semantics, see Example~\ref{expl:logics}.

Given $T$, the interpretation of the modalities is determined by
associating with each $\hearts\in\Lambda$ a predicate lifting
$\Sem{\hearts}$ for $T$. Here, a \emph{predicate lifting} (for $T$) is
a family of maps $\lambda_X:\Pow{X} \to\Pow{TX}$, where $X$ ranges
over all sets, satisfying the \emph{naturality} condition
\begin{equation*}
  \lambda_X(f^{-1}[A])=(Tf)^{-1}[\lambda_Y(A)]
\end{equation*}
for all $f:X\to Y$, $A\in\Pow{Y}$.  In other words, a predicate
lifting is a natural transformation
\begin{equation*}
  \lambda:\contrapow\to\contrapow\circ T^\Op
\end{equation*}
where $\contrapow:\Set^\Op\to\Set$ denotes the contravariant powerset
functor, given by $QX$ being the powerset of $X$ and $Qf(A)=f^{-1}[A]$
for $f:X\to Y$ and $A\subseteq Y$. The idea is that a predicate
lifting $\lambda_X$ converts a predicate on the set $X$ of states into
a predicate on the set $TX$ of structured collections over $X$. The
basic example, for $TX=\Pow(X)\times\Pow(U)$ as above, is
\begin{equation*}
  \Sem{\Box}_X(A)=\{(B,P)\mid B\subseteq A\},
\end{equation*}
which induces precisely the usual semantics of the box when composed
with taking preimages under the structure map, as in the clause for
the semantics of modalities given below. Given a predicate lifting
$\Sem{\hearts}$, we also have a predicate lifting $\Sem{\dual\hearts}$
for the dual modality $\dual\hearts$, given by 
\begin{equation*}
  \Sem{\dual\hearts}_X(A)=X-\Sem{\hearts}_X(X-A).
\end{equation*}
E.g.\ for $\Box$ as above and $\Diamond=\dualBox$, we have
$\Sem{\Diamond}_X(A)=\{(B,P)\mid B\cap A\neq\emptyset\}$.

As we work with fixpoints, we insist that \emph{all modal operators
  are monotone}, i.e.\ $\Sem{\hearts}_X:\Pow(X)\to\Pow(TX)$ is
monotone w.r.t.\ set inclusion for each $\hearts\in\Lambda$.  The
semantics of a formula $\phi$ with recursion variable $x$ (no other
variables will ever be evaluated in unsubstituted form) is a subset
$\Sem{\phi}_{(X,\xi)}(B)\subseteq X$, depending on a $T$-coalgebra
$(X,\xi)$ and a set $B\subseteq X$ serving as the interpretation
of~$x$.  The semantics of formulas $\phi$ without variables (in
particular of $\sharp$-formulas) will not depend on~$B$ and hence
will be denoted just by $\Sem{\phi}_{(X,\xi)}$; we write
$x\models_{(X,\xi)}\phi$ for $x\in\Sem{\phi}_{(X,\xi)}$. The set
$\Sem{\phi}_{(X,\xi)}(B)$ is defined by recursion over $\phi$:
\begin{align*}
  \Sem{x}_{(X,\xi)}(B)& =B\\
  \Sem{\neg\phi}_{(X,\xi)}(B) &= X\setminus\Sem{\phi}_{(X,\xi)}(B)\\
  \Sem{\phi\land\psi}_{(X,\xi)}(B) &= \Sem{\phi}_{(X,\xi)}(B)\cap \Sem{\psi}_{(X,\xi)}(B)\\
  \Sem{\hearts\phi}_{(X,\xi)}(B) & 
    =  \xi^{-1}\Sem{\hearts}_X(\Sem{\phi}_{(X,\xi)}(B))\\
  \Sem{\sharp_\gamma\phi}_{(X,\xi)}& =\bigcap\{B\subseteq X\mid
     \Sem{\gamma(\phi)}_{(X,\xi)}(B)\subseteq B\}\\
\end{align*}
The clause for $\sharp_\gamma\phi$ just says that
$\Sem{\sharp_\gamma\phi}_{(X,\xi)}$ is the least fixpoint of the map
$\Sem{\gamma(\phi)}_{(X,\xi)}:\Pow(X)\to\Pow(X)$, which is monotone
because all modalities are monotone and modal fixpoint schemes are
positive in the recursion variable. \emph{We fix the data $T$,
  $\Lambda$, $\Gamma$ etc.\ throughout.}

We recall that given $T$-coalgebras $(X,\xi)$ and $(Y,\zeta)$, a
\emph{$T$-coalgebra morphism} $f:(X,\xi)\to(Y,\zeta)$ is a map
$f:X\to Y$ making the diagram
\begin{equation*}
  \xymatrix{X\ar[d]_{\xi}\ar[r]^{f} & Y\ar[d]^{\zeta}\\
  TX \ar[r]_{Tf} & TY}
\end{equation*}
commute. Flat fixpoint formulas are invariant under coalgebra morphisms:
\begin{lemma}\label{lem:mor-preserve}
  Let $f:(X,\xi)\to(Y,\zeta)$ be a $T$-coalgebra morphism, let
  $B\subseteq Y$, and let $\phi$ be a flat fixpoint formula. Then
  \begin{equation}
    \Sem{\phi}_{(X,\xi)}(f^{-1}[B]) = f^{-1}[\Sem{\phi}_{(Y,\zeta)}(B)]. \label{eq:preserve}
  \end{equation}
\end{lemma}
(The lemma holds more generally for the full coalgebraic
$\mu$-calculus, with essentially the same proof; we refrain from
stating it in more generality here to avoid introducing additional
notation.)
\begin{proof}
  Induction over $\phi$; the Boolean cases are trivial, and the modal
  cases are by naturality of predicate liftings
  (cf.~\cite{Pattinson04}). We are left with the fixpoint case; we
  work with $\mu$-calculus notation, i.e.\ our remaining case is of
  the form $\mu x.\,\phi$, a closed formula because there is only one
  recursion variable. That is, we are to show that
  \begin{equation}
    \label{eq:preserve-fp}
    \Sem{\mu x.\,\phi}_{(X,\xi)} = f^{-1}[\Sem{\mu x.\,\phi}_{(Y,\zeta)}].
  \end{equation}
  It is well-known that we can approximate least fixpoints of monotone
  functions from below using ordinal-indexed chains. Specifically,
  $\Sem{\mu x.\,\phi}_{(X,\xi)}$ is the union of the sets
  $\Sem{\phi}_{(X,\xi)}^\alpha(\emptyset)$ indexed over ordinals
  $\alpha$, defined by
  $\Sem{\phi}_{(X,\xi)}^{0}(\emptyset)=\emptyset$, by
  $\Sem{\phi}_{(X,\xi)}^{\alpha+1}(\emptyset)=\Sem{\phi}_{(X,\xi)}(\Sem{\phi}_{(X,\xi)}^{\alpha}(\emptyset))$
  in the successor step, and by
  $\Sem{\phi}_{(X,\xi)}^{\alpha}(\emptyset)=\bigcup_{\beta<\alpha}\Sem{\phi}_{(X,\xi)}^{\beta}(\emptyset)$
  in the limit step; analogously, $\Sem{\mu x.\,\phi}_{(Y,\zeta)}$ is
  approximated from below by sets
  $\Sem{\phi}_{(Y,\zeta)}^\alpha(\emptyset)$. Since taking preimages
  under $f$ commutes with $\Sem{\phi}$ by the inductive hypothesis,
  and generally commutes with (infinite) unions and preserves
  $\emptyset$, an easy transfinite induction shows that
  \begin{equation*}
    \Sem{\phi}_{(X,\xi)}^\alpha(\emptyset)=f^{-1}[\Sem{\phi}_{(Y,\zeta)}^\alpha(\emptyset)]
  \end{equation*}
  for all $\alpha$. The inductive claim follows by forming the union
  over all $\alpha$ on both sides, again using commutation of
  preimages with unions. \rightqed
\end{proof}

 \begin{example}\label{expl:logics}
   We discuss select examples covered by the coalgebraic approach,
   starting with a more detailed exposition of the basic example of
   Kripke semantics and then moving on beyond. More examples are
   found, e.g.,
   in~\cite{SchroderPattinson09a,PattinsonSchroder10}. For the sake of
   readability, we elide the treatment of propositional atoms in all
   examples except Kripke semantics, as the technicalities are the
   same in all cases. 
  \begin{longenum}
  \item\label{item:K} \emph{Kripke semantics:} In mild generalization
    of the basic example discussed above, fixpoint extensions of the
    modal logics $K_m$ have modal operators $\Box_i$ for
    $i=1,\dots,m$, interpreted over the functor $T$ given on sets by
    $TX=(\Pow X)^m\times\Pow U$ using the predicate liftings
    $\lsem \Box_i \rsem_X(A) = \lbrace (B_1,\dots,B_m,P) \in
    (\Pow X)^m\times\Pow U \mid B_i \subseteq A \rbrace$.
    Coalgebras for $T$ are in 1-1 correspondence with $m$-relation
    Kripke models, and $\Sem{\Box_i}$ captures the usual semantics of
    the box operators. Atomic propositions are modelled as nullary
    modalities $a\in\Lambda$, interpreted by nullary predicate
    liftings
    \begin{equation*}
      \Sem{a}_X=\{(B_1,\dots,B_m,P) \in
      (\Pow X)^m\times\Pow U \mid a\in P\}. 
    \end{equation*}
    CTL, $*$-nesting-free PDL, and the logic of common knowledge all
    are flat fixpoint logics in this setting. The CTL operators have
    been exemplified above; the operator $[a^*]$ of PDL is the
    greatest fixpoint operator
    \begin{equation*}
      [a^*]=\flat_{p\land[a]x};
    \end{equation*}
    and the common knowledge operator $C$ is the greatest fixpoint
    operator $\flat_\gamma$ for
    \begin{equation*}
     \gamma=p\land\Land_{i=1}^m \Box_i x.
    \end{equation*}
    Although our main result does not support completely arbitrary
    fixpoint operators, it does (like
    already~\cite{SantocanaleVenema10}) cover operators that go beyond
    CTL. E.g.\ the operator `on every path, every even state satisfies
    \dots', which is not expressible in CTL~\cite{Emerson90,Wolper83},
    is the greatest fixpoint operator
    \begin{equation*}
      \flat_{p\land\Box\Box x}.
    \end{equation*}
    
    Strictly speaking one should note that so far, we have not covered
    CTL in the standard sense, according to which models are assumed
    to be \emph{serial}, i.e.\ every state is required to have at
    least one successor. We model this requirement coalgebraically by
    replacing $T$ as above with the functor
    $T'X=\Pow^* X\times\Pow U$ where $\Pow^*$ is the \emph{non-empty
      powerset functor} $\Pow^*$, i.e.~$\Pow^* X$ is the set of
    non-empty subsets of $X$. (Since CTL standardly uses only one
    next-step modality $AX$, we omit the exponent $m$.) We distinguish
    the arising fixpoint logics by the adjective \emph{serial}, and
    refer to absence of the seriality requirement by the adjective
    \emph{non-serial}.

    Further restricting the semantics to require that every state has
    \emph{exactly} one successor in fact produces the usual
    discrete-linear-order semantics of LTL (where $\Box$ is typically
    denoted $\ocircle$). Coalgebraically, this is reflected in using
    the functor $T''X=X\times\Pow U$; we refer to $T''$-coalgebras as
    \emph{deterministic} Kripke models. (For purposes of model
    checking, LTL is more generally interpreted over labelled
    transition systems, but the more restrictive discrete-linear-order
    semantics is equivalent for purposes of satisfiability.)
%     Moreover, it includes also more
%  complicated operators such as the $\sharp$-operator for
%  $p\land\Box\Diamond x\land\Diamond\Box x$. 
  \item\label{item:graded} \emph{Graded fixpoint logics} are sublogics
    of the \emph{graded $\mu$-calculus}~\cite{KupfermanEA02}.  They
    have modal operators $\gldiamond{k}$ `in more than $k$
    successors', with duals
    $\glbox{k}=\dual{\mystrut{4pt}\gldiamond{k}}$ `in all but $k$
    successors'. We interpret them over the functor $\Bag$ that takes
    a set $X$ to the set
    \begin{equation*}
      \Bag X=X\to\Nat\cup\{\infty\}
    \end{equation*}
    of multisets over $X$ (with possibly infinite multiplicities) by
    \begin{equation*}\textstyle
      \Sem{\gldiamond{k}}_X(A)=\{\mu\in\Bag X\mid \mu(A) >
      k\}
    \end{equation*}
    where we use $\mu\in\Bag X$ as an $\Nat\cup\{\infty\}$-valued
    measure, i.e.\ write $\mu(A)=\sum_{x\in A}\mu(x)$. This captures
    the semantics of graded modalities over
    \emph{multigraphs}~\cite{DAgostinoVisser02}, which is equivalent
    to the more customary Kripke semantics~\cite{Fine72} w.r.t.\
    satisfiability of fixpoint formulas
    (Lemma~\ref{lem:graded-sem}). In description logic, graded
    operators are called \emph{qualified number
      restrictions}~\cite{BaaderEA03}.  The example mentioned
    in~\cite{KupfermanEA02}, a graded fixpoint formula expressing that
    the current state is the root of a finite binary tree all whose
    leaves satisfy $p$, can be expressed by the $\sharp$-operator for
    \begin{equation*}
      p\lor\gldiamond{1}x.
    \end{equation*}
    Similarly, the $\sharp$-operator for
    \begin{equation*}
      p\lor\glbox{k}x
    \end{equation*}
    expresses that $p$ holds somewhere on every infinite $k+1$-ary
    tree starting at the current state. To add an example where the
    recursion variable $x$ appears under more than one modality, the
    $\sharp$-operator for
    \begin{equation*}
      p\lor\gldiamond{1}\gldiamond{1}x
    \end{equation*}
    expresses that the current state is the root of a finite binary
    tree all whose leaves are at even distance from the root and
    satisfy $p$.
  \item \emph{Probabilistic fixpoint logics}, i.e.\ fixpoint
    extensions of probabilistic modal
    logic~\cite{LarsenSkou91,FaginHalpern94,HeifetzMongin01}, have modal operators
    $L_r$ `in the next step, it holds with probability at least $r$
    that', for $r\in[0,1]\cap\Rat$. They are interpreted over the
    functor $\PDist$ that maps a set $X$ to the set of discrete
    probability distributions on $X$ by putting
    \begin{equation*}
      \Sem{L_r}_X(A)=\{P\in\PDist X\mid PA\geq r\}.
    \end{equation*}
    Coalgebras for $\PDist$ are Markov chains. Flat probabilistic
    fixpoint logics in this sense are fragments of the
    \emph{probabilistic $\mu$-calculus} in the sense introduced by
    Cîrstea et al.~\citeyear{CirsteaEA11} (to be distinguished from
    the $[0,1]$-valued logic of the same
    name~\cite{MorganMcIver97,HuthKwiatkowska97}) as an instance of
    the coalgebraic $\mu$-calculus. The probabilistic $\mu$-calculus
    was subsequently reinvented by Liu et al.~\citeyear{LiuEA15} (Liu
    et al.~prove that model checking the probabilistic $\mu$-calculus
    is in \textsc{UP}$\cap$\textsc{coUP} and satisfiability is in
    2\ExpTime while Cîrstea et al.~already show that satisfiability is
    \ExpTime-complete). More generally, one can admit linear
    inequalities between probabilities, as, e.g., in work by Fagin and
    Halpern~\citeyear{FaginHalpern94}, as long as one pays attention
    to monotonicity; this allows, e.g., for statements of the form
    `more probably than not'.

    In a view of probabilistic logic as a logic of reactive systems,
    we can use the $\flat$-operator
    \begin{equation*}
      AG_r=\flat_{p\land L_r x}
    \end{equation*}
    to express formulas like $AG_r\,\neg\mathsf{fail}$, stating that
    the system will, at any point during its run time, fail with
    probability at most $1-r$; a sensible specification for systems
    that may sometimes fail but should not fail excessively often (as
    announced, we silently include propositional atoms such as
    $\mathsf{fail}$ in the syntax). 
    
    Alternatively, we may interpret the operators $L_p$
    epistemically. We extend the logic to multiple agents in the same
    way as for $K$ in Item~\ref{item:K}, obtaining a logic with
    probabilistic operators $L^i_r$ read `agent $i$ believes with
    confidence $r$ that'. We then have an uncertain variant $C_r$ of
    the common knowledge operator, namely the $\flat$-operator for
    \begin{equation*}
      \textstyle\Land_{i=1}^m L_r^i(p\land x).
    \end{equation*}
    Thus, $C_r\phi$ is read `everyone believes with confidence $r$
    that $\phi$ holds and that everyone believes with confidence $r$
    that $\phi$ holds etc.', in short `$\phi$ is commonly believed
    with confidence $r$'. 

    A variant is the operator $C_{r,q}$ which, when applied to a
    formula $\phi$, separates belief in~$\phi$ from beliefs about
    other agents: $C_{r,q}$ is the $\flat$-operator for
    \begin{equation*}
      \textstyle\Land_{i=1}^n (L_r^ip\land L_q^i x),
    \end{equation*}
    and thus $C_{r,q}\phi$ states that everyone believes $\phi$ with
    confidence $r$, and believes with confidence $q$ that all agents
    believe the same, etc. (Note that $C_{r,r}$ is not the same
    as~$C_r$!)

  \item \emph{Conditional fixpoint logics} have a single binary modal
    operator $\CondArrow$, written in infix notation. The intended
    reading of $a\CondArrow b$ is `if $a$, then \emph{normally}
    $b$'. Conditional logics come with a wide variety of
    axiomatizations and semantics~(see,
    e.g.,~\cite{PattinsonSchroder10} for an overview). E.g., the
    minimal conditional logic $\CK$ is interpreted over the functor
    $\CF$ that maps a set $X$ to the set $\Pow X\to\Pow X$ (more
    precisely $\contrapow X\to\Pow X$, where $\contrapow$ is
    contravariant powerset), whose coalgebras are selection function
    models~\cite{Chellas80}, by putting
    \begin{equation*}
      \Sem{\CondArrow}_X(A,B)= \{f\in\CF X\mid f(A)\subseteq B\}.
    \end{equation*}
    (Thus, $\CondArrow$ is monotone in the second but not antimonotone
    in the first argument, which is precisely the point of inventing
    it. Technically, this is fine as long as we form fixpoints only
    over the second argument, i.e.\ let the recursion variable $x$
    appear in $\gamma\in\Gamma$ only to the right of $\CondArrow$ as
    in the example below.)

    We can combine coalgebraic logics freely using results
    of~\cite{SchroderPattinson07}. E.g., combining conditional logic
    with multi-agent $K$ in its description logic incarnation
    $\mathcal{ALC}$~\cite{BaaderEA03}, we can define an abstract
    concept of animal taxa with two-gender descendancy as
    $\flat_\gamma(\mathsf{FirstOfItsKind},\mathsf{Male},\mathsf{Female})$
    with $\gamma$ given as
    \begin{equation*}
      (\neg p_1)\Rightarrow((\exists\,\mathsf{hasParent}.\,p_2\sqcap x)
      \sqcap(\exists\,\mathsf{hasParent}.\,p_3\sqcap x)).
    \end{equation*}
    A taxon is a fixpoint of
    $\gamma(\mathsf{FirstOfItsKind},\mathsf{Male},$ $\mathsf{Female})$
    if all individuals that do not belong to some assumed first
    ancestor generation of the species \emph{normally} have two
    parents of the same species, and the greatest fixpoint subsumes
    all animals belonging to such taxa. The use of the conditional
    $\Rightarrow$ instead of standard material implication $\modimpl$
    takes into account that these days, a given exceptional sheep
    might, e.g., be a clone and thus not have parents in the strict
    sense.
  \item\label{item:amc}\emph{The alternating-time $\mu$-calculus
      (AMC)}~\cite{AlurEA02} has modal operators
    $\ATLdiamond{C}\bigcirc$ read `coalition $C$ of agents can enforce
    \dots\ in one step', where a \emph{coalition} $C$ is a subset
    $C\subseteq N$ of a fixed set $N=\{1,\dots,n\}$ of agents; we
    shall also write $[C]$ in place of $\ATLdiamond{C}\bigcirc$ as in
    coalition logic~\cite{Pauly02}. The semantics of the coalitional
    modalities is defined over \emph{concurrent game structures} (or
    \emph{game frames}), and can be captured
    coalgebraically~\cite{SchroderPattinson09a}: We define a functor
    $\Gm$ by
    \begin{equation*}
      \Gm  X = \lbrace (f, (k_i)_{i \in N}) \mid\textstyle f:
      \big(\prod_{i \in N}\big) [k_i] \to X \rbrace
    \end{equation*}
    where $k_i\in\Nat$ and $[k_i]=\{1,\dots,k_i\}$, and by
    $\Gm g(f, (k_i)_{i \in N})=(g\circ f, (k_i)_{i \in N})$ for
    $g:X\to Y$. This captures a form of concurrent game where each
    agent $i\in N$ chooses a move $j_i\in[k_i]$ and the joint choice
    determines an \emph{outcome} $f(j_1,\dots,j_n)\in X$.  (The
    semantics given by Pauly~\citeyear{Pauly02} differs slightly in
    that the agents can have unrestricted sets of available moves
    rather than only finite ones.)  Coalgebras $(X,\xi)$ for $\Gm$ are
    \emph{concurrent game structures}~\cite{AlurEA02}; they associate
    to each state $x\in X$ a concurrent game $\xi(x)\in\Gm X$ whose
    outcomes are states, and thus allow for plays with multiple
    successive moves.  The semantics of the modalities $[C]$ is given
    by the liftings
    \begin{equation*}
      \lsem [C] \rsem_X(A) = \lbrace (f, (k_i)_{i \in N}) \in \Gm  X \mid
      \exists (j_i\in[k_i])_{i\in C}.\,\forall (j_i\in[k_i])_{i \in N
        \setminus C}.\, f((j_i)_{i \in N}) \in A \rbrace.
    \end{equation*}
    That is, a state $x$ in a concurrent game structure satisfies
    $[C]\phi$ if the agents in $C$ have a joint choice of moves such
    that regardless of the choice of moves by the other agents, the
    outcome satisfies $\phi$.

    One of the flat fragments of the AMC is Alternating-Time Temporal
    Logic (ATL)~\cite{AlurEA02}. E.g., the ATL-operator
    $\ATLdiamond{C}p_1\until p_2$, read `coalition $C$ can eventually
    force $p_2$ and meanwhile maintain $p_1$', is the
    $\sharp$-operator for
    \begin{equation*}
      p_2\lor(p_1\land [C] x).
    \end{equation*}
    As already in the case of CTL, flat fixpoints in the AMC go
    considerably beyond ATL; e.g.\ the $\flat$-operator for
    $p\land[\emptyset][\emptyset] x$ (`$p$ holds in all even states
    along any path') is not even in the more expressive logic
    ATL$^*$~\cite{AlurEA02,Dam94}. A similar flat operator, the
    $\flat$-operator for $[C] (p\land [D](q\land x))$, expresses that
    coalitions $C$ and $D$ can forever play ping-pong between $p$ and
    $q$.
  \item \emph{Monotone fixpoint logics} have a modal operator $\Box$,
    interpreted over the \emph{monotone neighbourhood functor} defined
    by
    \begin{equation*}
      \CM X=\{\mathfrak{A}\in\Pow(\Pow X)\mid \mathfrak{A}\text{
        upwards closed}\}
    \end{equation*}
    by means of the predicate lifting
    \begin{equation*}
      \Sem{\Box}_X(A)=\{\mathfrak{A}\in\CM X\mid A\in\mathfrak{A}\}.
    \end{equation*}
    (The functor $\CM$ acts on maps $f:X\to Y$ by
    $\CM f(\mathfrak{A})=\{B\in\Pow Y\mid
    f^{-1}[B]\in\mathfrak{A}\}$,
    and hence is a subfunctor of the double contravariant powerset
    functor.) Often, the axioms $\Box\top$ and $\Diamond\top$ are
    imposed where $\Diamond=\dualBox$ denotes the dual of $\Box$.
    This amounts to using the subfunctor $\CM_s$ of $\CM$ given by
    \begin{equation*}
      \CM_s X=\{\mathfrak{A}\in\CM X\mid \emptyset\notin\mathfrak{A}\owns X\},
    \end{equation*}
    the \emph{serial monotone neighbourhood functor}, whose coalgebras
    are \emph{serial monotone neighbourhood frames}. In particular,
    these form the semantic setting of concurrent PDL~\cite{Peleg87}
    and Parikh's game logic~\cite{Parikh85}, where additionally
    operators are indexed over atomic programs or games, respectively;
    this is modelled coalgebraically in the same way as multi-modal
    $K$ (Item~\ref{item:K}). The $*$-nesting-free fragments of
    concurrent PDL and game logic are flat fixpoint logics. E.g., game
    logic has operators $\langle\gamma\rangle$, indexed over composite
    games $\gamma$ and read `Angel has a strategy to enforce \dots in
    game $\gamma$'. Games are formed from atomic games using the usual
    constructs for regular expressions as in PDL, and additionally the
    dualizing operator $(-)^d$ which swaps the roles of the players in
    a game. E.g.\ one has a demonic iteration operator $(-)^\times$
    defined by $\gamma^\times=((\gamma^d)^*)^d$. The formula
    $\langle\gamma^\times\rangle\phi$ thus says that Angel has a
    strategy to enforce $\phi$ in the game where $\gamma$ is played
    repeatedly, with Demon choosing the number of rounds. If $\gamma$
    is star-free (i.e.\ contains neither $(-)^*$ nor $(-)^\times$), we
    can phrase $\langle\gamma^\times\rangle$ as the $\flat$-operator
    for $p\land\langle\gamma\rangle x$.
  \end{longenum}
\end{example}
It remains to show the mentioned equivalence of the Kripke semantics
and the multigraph semantics of the graded $\mu$-calculus
(Example~\ref{expl:logics}.\ref{item:graded}), generalizing the
equivalence for the fixpoint-free case~\cite{Schroder07}.
\begin{lemma}
  \label{lem:graded-sem}
  A formula in the flat graded $\mu$-calculus is satisfiable over
  (finite) Kripke frames iff it is satisfiable over (finite)
  multigraphs.
\end{lemma}
\begin{proof}
  For the sake of simplicity, we continue to elide propositional
  atoms. Since every Kripke frame can be regarded as a multigraph,
  `only if' is clear. To show `if', let $\phi$ be a flat graded
  fixpoint formula, and let $x_0$ be a state in a $\Bag$-coalgebra
  $(X,\xi)$ such that $x_0\models_{(X,\xi)}\phi$. Let $k_0$ be maximal
  such that $\gldiamond{k_0}$ occurs in $\phi$.  Observe that $\phi$
  remains satisfied if we replace $\xi$ with $\xi'$ where
  $\xi'(x)(y)=\min(\xi(x)(y),k_0+1)$ (formally, this is proved by
  induction on $\phi$), so we can assume that all $\xi(x)(y)$ are
  finite (in fact, at most $k_0+1$). Now construct a Kripke model
  $(\bar X,R)$ by making sufficiently many copies of states, as in
  \cite[Remark~6]{Schroder07}: take as states in $\bar X$ all pairs
  $(y,j)\in X\times\Nat$ such that $\xi(x)(y)>j$ for some~$x$, and in
  this case put $(x,i)R(y,j)$ for all $i$ such that $(x,i)\in\bar X$.
  Note that $\bar X$ is finite if $X$ is finite. Like for any Kripke
  frame, we can equivalently regard $(\bar X,R)$ as a multigraph
  $(\bar X,\bar\xi)$ where $\bar\xi(x,i)(y,j)=1$ if $(x,i)R(y,j)$, and
  $\bar\xi(x,i)(y,j)=0$ otherwise. Let $\pi:\bar X\to X$ denote the
  projection that maps $(x,i)$ to $x$. By
  Lemma~\ref{lem:mor-preserve}, it suffices to show that
  \begin{equation*}
    \pi:(\bar X,\bar\xi)\to(X,\xi)
  \end{equation*}
  is a $\Bag$-coalgebra morphism, i.e.\ that
  \begin{equation*}
    \Bag\pi(\bar\xi(x,i))=\xi(x).
  \end{equation*}
  Indeed the multiplicity of $y\in X$ in the multiset on the left hand
  side is the cardinality of the set
  $\{(y,j)\in\bar X\mid (x,i)R(y,j)\}$, which equals $\xi(x)(y)$ by
  construction of $R$. \rightqed
\end{proof}
\section{The Generic Axiomatization}\label{sec:ax}

\noindent The generic semantic and syntactic framework of the previous
section comes with a generic, parametrized \emph{deduction system},
whose completeness will be our main result. We begin with the fixed
part of the deduction system. We include full \emph{propositional
  reasoning}, i.e.\ introduction of substituted propositional
tautologies and modus ponens. Fixpoints are governed by the obvious
\emph{Kozen-Park axiomatization}: we have the \emph{unfolding} axiom
\begin{equation*}
  \gamma(\phi,\sharp_\gamma\phi)\to\sharp_\gamma\phi
\end{equation*}
and the \emph{fixpoint induction} rule
\begin{equation*}
%  \infrule{\gamma(\phi,\chi)\modimpl\chi}{\sharp_\gamma(\phi)\modimpl\chi},
\infrule{
\gamma(\phi,\chi)\modimpl\chi
}{
\sharp_\gamma\phi\modimpl\chi
}
,
\end{equation*}
for all formulas $\phi,\chi$; together, these axioms capture the fact
that $\sharp_\gamma\phi$ is the least prefixpoint of
$\gamma(\phi)$. 
\begin{lemma}\label{lem:fp-mon}
  The \emph{monotonicity} and \emph{congruence} rules
  \begin{equation*}
    \infrule{\phi\to\psi}{\sharp_\gamma\phi\to\sharp_\gamma\psi}
    \qquad
    \infrule{\phi\modiff\psi}{\sharp_\gamma\phi\modiff\sharp_\gamma\psi}
  \end{equation*}
  are derivable.
\end{lemma}
\begin{proof}
  The congruence rule is derivable from the monotonicity rule. To
  derive the latter, assume $\phi\to\psi$. Since $\gamma$ is positive
  in the parameter variable, we can then derive
  $\gamma(\phi,\sharp_\gamma\psi)\to\gamma(\psi,\sharp_\gamma\psi)$. By
  unfolding, we derive
  $\gamma(\phi,\sharp_\gamma\psi)\to\sharp_\gamma\psi$, and then
  $\sharp_\gamma\phi\to\sharp_\gamma\psi$ by fixpoint
  induction. \rightqed
\end{proof}
The variable part of the proof system is the axiomatization of the
modal operators, which turns out to be completely orthogonal to the
fixpoint axiomatization. In fact, we can just reuse complete rule sets
for the purely modal part of the
logic~\cite{Schroder07,SchroderPattinson09a}.  First some notation.

\begin{definition}\label{def:propstuff}
  We denote the set of propositional formulas over a set~$Z$ by
  $\Prop(Z)$, % , its subset of positive propositional formulas (formed
  % using only $\land$ and $\lor$) over~$Z$ by $\Pos(Z)$
  and the set $\{\hearts a\mid \hearts\in\Lambda,a\in Z\}$ by
  $\Lambda(Z)$.  A \emph{literal} over~$Z$ is either an element of~$Z$
  or the negation of such an element, i.e.\ has the form $\epsilon z$
  where $z\in Z$ and $\epsilon\in\{\cdot,\neg\}$ is either nothing or
  negation.  A (disjunctive) \emph{clause} is a finite (possibly
  empty) disjunction of literals; a \emph{conjunctive clause} is a
  finite conjunction of literals. A clause (disjunctive or
  conjunctive) is \emph{contracted} if it contains every literal at
  most once. For $\phi,\psi\in\Prop(Z)$, we say that~$\phi$
  \emph{propositionally entails}~$\psi$, and write $\phi\PLentails\psi$,
  if $\phi\modimpl\psi$ is a propositional tautology.  Similarly,
  $\Phi\subseteq\Prop(Z)$ propositionally entails $\psi$
  ($\Phi\PLentails\psi$) if there exist $\phi_1,\dots,\phi_n\in\Phi$
  such that $\phi_1\land\dots\land\phi_n\PLentails\psi$. We write $2$
  for the set $\{\bot,\top\}$ of truth values.
  %This induces
  %obvious notions of propositional consistency and propositional
  %equivalence.
  For $\phi\in\Prop(Z)$, we denote the evaluation of $\phi$ in the
  Boolean algebra $\Pow X$ under a valuation $\tau: Z \to \Pow X$ by
  $\Sem{\phi}_{X,\tau}$, and write $X, \tau \models \phi$ if
  $\Sem{\phi}_{X,\tau}=X$. For $\psi\in\Prop(\Lambda(\Prop(Z)))$, the
  interpretation $\lsem \psi \rsem_{TX, \tau}$ of $\psi$ in the
  Boolean algebra $\Pow(TX)$ under $\tau$ is the inductive extension
  of the assignment
  $\lsem \hearts (z) \rsem_{TX, \tau} = \lsem \hearts \rsem_X\tau(z)$.
  We write $TX, \tau \models \psi$ if
  $\lsem \psi \rsem_{TX, \tau} = TX$. A propositional formula over
  $\Lambda(V)$ is \emph{injective} if it mentions every variable at
  most once.  
\end{definition}
\noindent We can now give the formal definition of the modal rule
format. 
\begin{definition}\label{def:os-rules}
  A \emph{one-step rule} $R=\phi/\chi$ consists of a \emph{premise}
  $\phi\in\Prop(V)$ and a \emph{conclusion} $\chi$ which is an
  injective (disjunctive) clause over $\Lambda(V)$ (recall that $V$ is
  the set of variables), where every variable in $\phi$ appears also
  in~$\chi$.  We say that $R$ is \emph{monotone} (a notion similar to
  one introduced by Cirstea et al.~\citeyear{CirsteaEA11} for rules
  phrased without negation) if whenever $\chi$ is positive (negative)
  in a variable $a\in V$ then $\phi$ is positive (negative) in $a$.
  The rule $R$ is \emph{one-step sound} if whenever
  $X,\tau\models\phi$ for a valuation $\tau:V\to\Pow X$, then
  $TX,\tau\models\chi$. A set $\Rules$ of one-step rules is
  \emph{one-step complete} if, whenever $TX,\tau\models\psi$ for a
  set~$X$, an injective clause $\psi$ over $\Lambda(V)$, and a
  $\powerset X$-valuation~$\tau$, then~$\psi$ is % \emph{(cutfree)
  % provable} over $X,\tau$, i.e.\ propositionally entailed by clauses
  % (a contracted clause) $\chi\sigma$ where $\phi/\chi\in\Rules$
  % and~$\sigma$ is a $\Prop(V)$-substitution such that
  % $X,\tau\models\phi\sigma$.
  \emph{provable over $X,\tau$}, i.e.\ propositionally entailed by
  clauses $\chi\sigma$ where $\phi/\chi\in\Rules$ and~$\sigma$ is a
  $\Prop(V)$-substitution such that $X,\tau\models\phi\sigma$.
  Moreover, $\Rules$ is \emph{one-step cutfree complete} if, whenever
  $TX,\tau\models\psi$ for $X,\tau,\psi$ as above, then~$\psi$ is
  \emph{cutfree provable over $X,\tau$}, i.e.\ $\chi\PLentails\psi$
  for some $\phi/\chi\in\Rules$ such that $X,\tau\models\phi$.
\end{definition}
\begin{remark}
  In the terminology of~\cite{SchroderPattinson09a}, one-step cutfree
  complete rule sets correspond to one-step complete rule sets which
  are closed under contraction, resolution, and injective renamings of
  the propositional variables. Notice in particular that $\chi$ as in
  the definition of cutfree provability is injective, being the
  conclusion of a one-step rule, and hence contracted.
% In the definition of
  % one-step cutfree completeness, one may equivalently restrict the
  % substitution $\sigma$ to be a $V$-substitution.
\end{remark}

% \begin{lemma}\label{lem:osc-neg}
%   Let $X$ be a set, let $\tau:V\to\Pow X$, and suppose that the image
%   of $\tau$ is closed under complement. Then a conjunction $\psi$ over
%   $\Lambda(V)$ is one-step cut-free $\tau$-consistent iff whenever
%   $\phi/\chi\in\Rules$ and $\sigma:V\to V$ is a renaming such that
%   $\chi\sigma$ is contracted and $\psi\PLentails\neg\chi\sigma$, then
%   $\Sem{\neg\phi\sigma}_{X,\tau}\neq\emptyset$.
% \end{lemma}
% \begin{proof}
%   By assumption, we have a renaming $n:V\to V$ such that
%   $\tau(n(a))=X-\tau(a)$. 

%   \emph{`Only if':} Let $\phi/\chi$ and $\sigma$ be as in the
%   claim. We have $\neg\chi\

%  Given that the image of $\tau$ is closed under complement,
%   we can use renamings $V\to V$ to emulate negation of variables and
%   hence convert back and forth between negations $\neg\chi$,
%   $\neg\phi$ and duals $\dual\chi$, $\dual\phi$.
% \end{proof}

%
\noindent As the last
parameter of the framework, we
\begin{quote}
  \emph{fix from now on a one-step complete set $\Rules$ of one-step
    sound one-step rules, and denote the arising logic by
    $\Lang_\sharp$}. 
\end{quote}
Rules $\phi/\psi\in\Rules$ are applied in substituted form, i.e\ for
every substitution $\sigma$, we may conclude $\psi\sigma$ from
$\phi\sigma$. In summary, the proof system consists of propositional
reasoning, the unfolding axiom, the fixpoint induction rule, and the
rules in $\Rules$. It is easy to see that this system is sound. As
usual, we write $\entails\phi$ if a formula $\phi$ is provable. We say
that $\phi$ is \emph{consistent} if $\neg\phi$ is not provable. One
should note that the system without the fixpoint rules (i.e.\
comprising only propositional reasoning and a one-step complete set of
modal rules) is complete for the fixpoint-free
language~\cite{Schroder07}.

We conclude the technical part of the section with some facts on the
relationship between the two notions of one-step completeness.
\begin{definition}
  A one-step rule $\phi/\chi$ is \emph{$\Rules$-derivable} if it is
  propositionally entailed by conclusions $\psi\sigma$ of rules
  $\rho/\psi\in\Rules$ where $\sigma$ is a $\Prop(V)$-substitution and
  $\phi\PLentails\rho\sigma$.
\end{definition}
\begin{lemma}\label{lem:oss-derive}
  All one-step sound one-step rules are $\Rules$-derivable.
\end{lemma}
\begin{proof}
  Let $\phi/\chi$ be one-step sound. Let $V_0\subseteq V$ be the set
  of propositional variables that occur in $\chi$, and put
  \begin{equation*}
    X=\{\kappa:V_0\to 2\mid \kappa\models\phi\}\qquad
    \tau(a)=\{\kappa\in X\mid\kappa(a)=\top\}.
  \end{equation*}
  Then $X,\kappa\models\phi$, so that $TX,\kappa\models\chi$ by
  one-step soundness. By one-step completeness of $\Rules$, $\chi$ is
  propositionally entailed by clauses $\psi\sigma$ where
  $\rho/\psi\in\Rules$ and $\sigma$ is a $\Prop(V)$-substitution such
  that $X,\tau\models\rho\sigma$. It suffices to show that
  $\phi\PLentails\rho\sigma$, which however is clear by construction
  of $X,\tau$. \rightqed
\end{proof}
\begin{lemma}\label{lem:mon-rules-complete}
  The set of $\Rules$-derivable monotone one-step rules is one-step cutfree
  complete.
\end{lemma}
\noindent In other words, \emph{a clause over $\Lambda(V)$ is provable
  over $X,\tau$ iff it is cutfree provable over $X,\tau$ using an
  $\Rules$-derivable monotone one-step rule}. We thus say that a conjunctive
clause $\rho$ over $\Lambda(V)$ is \emph{one-step $\tau$-consistent}
for $\tau:V\to\Pow X$ if $\neg\rho$ is not provable over $X,\tau$,
equivalently not cutfree provable over $X,\tau$ using an $\Rules$-derivable
monotone one-step rule.
\begin{proof}
  % By Lemma~\ref{lem:oss-derive}, this follows from the known fact that
  % the set of one-step sound one-step rules is one-step cutfree
  % complete in the sense that we use here~\cite[Proof of
  % Theorem~18]{Schroder07}.
  By Lemma~\ref{lem:oss-derive}, this follows once we show that the
  set of one-step sound monotone one-step rules is one-step cutfree
  complete. This is proved by Cirstea et
  al.~\citeyear[Proposition~4.7]{CirsteaEA11} for a formally even more
  restrictive class of rules (also called
  monotone\footnote{Specifically, a rule is monotone in the sense of
    Cirstea at al.\ if its premise is a positive formula and its
    conclusion is a disjunction of atoms of the form $\hearts a$ or
    $\dual\hearts a$; in fact, monotone rules in our sense can be
    transformed into this format by taking negation normal forms and
    substituting away negated variables.}).  \rightqed
\end{proof}
\begin{example}\label{expl:rules}
  One-step complete rule systems have been exhibited for all logics of
  Example~\ref{expl:logics} and many
  more~\cite{Pattinson03,CirsteaPattinson07,SchroderPattinson09a,PattinsonSchroder10,KupkePattinson10}. In
  some cases~\cite{CirsteaPattinson07,Pattinson03}, axiomatizations
  have been phrased in terms of one-step \emph{axioms}, i.e.\ formulas
  in $\phi\in\Prop(\Lambda(\Prop(V)))$ that can be introduced in
  substituted form $\phi\sigma$ with $\sigma$ a substitution of
  propositional variables by formulas; it has been shown that this
  format is interconvertible with one-step rules~\cite{Schroder07}. We
  recall some examples in more detail, converting to one-step rules
  where necessary.
  \begin{longenum}
  \item \emph{Kripke semantics:} The standard axiomatization of the
    modal logic $K_1$, with $\Box_1$ written as just $\Box$, can be
    phrased in terms of one-step rules for necessitation,
    monotonicity, and normality as
    \begin{equation*}
      \infrule{a}{\Box a}\quad\infrule{a\to b}{\Box a\to \Box b}\quad
      \infrule{a\land b\to c}{\Box a\land\Box b\to\Box c}
    \end{equation*}
    where here and in the following, we write clauses
    $\neg a_1\lor\dots\lor\neg a_n\lor b_1\lor\dots b_m$ as
    implications $a_1\land\dots\land a_n\to b_1\lor\dots\lor b_m$.
    Given this axiomatization, a one-step rule $\phi/\psi$ is
    $\Rules$-derivable iff $\psi$ contains disjuncts
    $\neg\Box a_1,\dots,\neg\Box a_n,\Box b$ (with $n\ge 0$) such that
    $\phi\PLentails a_1\land\dots\land a_n\to b$. 

    A more restricted semantics over serial Kripke models as for
    serial CTL (Example~\ref{expl:logics}.\ref{item:K}) is reflected
    in the additional rule
    \begin{equation*}
      \infrule{\neg a}{\neg\Box a}.
    \end{equation*}
    Restricting additionally to deterministic models as used for LTL
    (and switching from $\Box$ to $\ocircle$ as customary) logically
    corresponds to adding the rule
    \begin{equation*}
      \infrule{a\lor b}{\ocircle a\lor\ocircle b}
    \end{equation*}
    (which, given the other rules, is mutually interderivable with the
    better-known axiom $\neg\ocircle a\to\ocircle\neg a$).
    % As follows: rule => ax: show []a \/ -[]a from a \/ -a;
    % ax -> rule: a \/ b => -a -> b =(M)=> []-a -> []b 
    % =ax=> -[]a -> []b => []a \/ []b.
  \item\label{item:rules-graded} \emph{Graded fixpoint logics:} Rephrasing a known
    complete axiomatization of graded modal logic~\cite{DeCaro88}, we obtain
    the rules
      \begin{gather*}
    \lrule{RG1}{a\modimpl b}{\gldiamond{n+1}a\modimpl\gldiamond{n}b}
    \qquad
    \lrule{A1}
	  {c\modimpl a\lor b}
	  {\gldiamond{n_1+n_2}c\modimpl\gldiamond{n_1}a\lor\gldiamond{n_2}b}\\
    \lrule{A2}
	  {\begin{array}{c}
	a\lor b\modimpl c\\
	a\land b\modimpl d
    \end{array}}
	  {\gldiamond{n_1}a\land\gldiamond{n_2}b\modimpl
	    \gldiamond{n_1+n_2+1}c\lor\gldiamond{0}d}
	  \qquad
    \lrule{RN}{\lnot a}{\lnot\gldiamond{0}a}.
  \end{gather*}
  These rules are clearly one-step sound. They have previously been
  shown to be one-step complete~\cite{SchroderPattinson09a} by
  reference to the previous completeness proof for graded modal logic;
  we give a simple stand-alone proof in
  Lemma~\ref{lem:graded-completeness}.
\item \emph{Probabilistic fixpoint logics:} We can reuse the one-step
  complete rule set for probabilistic modal
  logic~\cite{HeifetzMongin01,CirsteaPattinson07}. For the extended
  language with linear inequalities on probabilities, one has the
  one-step cutfree complete rule set given by Kupke and
  Pattinson~\citeyear{KupkePattinson10}, noting that one-step cutfree
  complete rule sets can be restricted to any subset of the modal
  operators~\cite{SchroderPattinson09a}, in particular to monotone
  linear inequalities.
\item \emph{Conditional fixpoint logics:} one-step complete rule sets
  are known for various flavours of conditional
  logic~\cite{PattinsonSchroder10,SchroderEA10}.
\item \emph{Alternating-time $\mu$-calculus:} The following one-step
  complete set of rules is implicit in~\cite{Pauly02} (see
  also~\cite{SchroderPattinson09a}):
  \begin{equation*}
    \infrule{\lnot a}
    {\lnot\plbox{C} a}
    \quad
    \infrule{a}{\plbox{C}a}
    \quad
    \infrule{a \lor b}{\plbox{0}a \lor \plbox{N} b}
    \quad
    \infrule{a\land b \modimpl c}
    {[C]a \land [D]b \to [C \cup D]c}
  \end{equation*}
  where $C$ and $D$ are disjoint in the last rule. In words: no
  coalition can enforce the logically impossible; every coalition can
  enforce logical tautologies; either $a$ is unavoidable or $\neg a$
  can be enforced by all agents in collaboration; and disjoint
  coalitions can combine their abilities.
\item \emph{Monotone fixpoint logics:} When we interpret $\Box$ over
  the monotone neighbourhood functor, we have only the rule
  \begin{equation*}
    \infrule{a\to b}{\Box a \to\Box b}.
  \end{equation*}
  Seriality is captured by the additional rules
  \begin{equation*}
    \infrule{a}{\Box a}\qquad
    \infrule{\neg a}{\neg\Box a}.
  \end{equation*}    
  \end{longenum}
\end{example}
% formally, it expands into the positive formula
% \begin{equation*}
%   \Land_{\substack{J_1\subseteq\{1,\dots,n\},J_2\subseteq\{1,\dots,m\}\\
%     \sum_{i\in J_1}-r_i+\sum_{j\in J_2}s_j<0}}
% \;  \big(\Lor_{i\in J_1}a_i\lor\Lor_{j\notin J_2}b_j\big).
% \end{equation*}
% \noindent Another consequence of Lemma~\ref{lem:oss-derive} is
% \begin{lemma}\label{lem:contraction}
%   The set of $\Rules$-derivable one-step rules is \emph{closed under
%     contraction}~\cite{SchroderPattinson09a}. That is, if $\phi/\chi$
%   is $\Rules$-derivable and $\sigma:V\to V$ is such that for literals
%   $\epsilon_1\hearts_1a_1,\epsilon_2\hearts_2a_2$ in $\chi$ with
%   $\epsilon_i\in\{\cdot,\neg\}$, $\sigma(a_1)=\sigma(a_2)$ only if
%   $\epsilon_1\hearts_1=\epsilon_2\hearts_2$, then there exists an
%   $\Rules$-derivable one-step rule $\phi\sigma/\chi'$ such that $\chi'$ is
%   propositionally equivalent to $\chi\sigma$.
% \end{lemma}
% \noindent (Note that $\phi'/\chi'$ being a one-step rule entails that
% $\chi'$ is injective, in particular contracted.)
% \begin{proof}
%   Construct $\chi'$ by removing duplicate literals from
%   $\chi\sigma$. By the assumption on $\sigma$, $\chi'$ is
%   injective. Since $\phi/\chi$ is one-step sound, so is
%   $\phi\sigma/\chi'$, which is thus $\Rules$-derivable by
%   Lemma~\ref{lem:oss-derive}.
% \end{proof}
\begin{remark}\label{rem:finite-ax}
  We point out that in all examples with finite modal similarity
  type~$\Lambda$, the rule sets given above are finite, so that our
  completeness result will establish finite axiomatizability; this
  holds in particular for alternating-time logics. When $\Lambda$ is
  infinite, we cannot reasonably expect a finite axiomatization. The
  rules for graded modalities are \emph{locally finite} in the sense
  that every modality is mentioned only in finitely many axioms. The
  rules for probabilistic
  logic~\cite{HeifetzMongin01,CirsteaPattinson07} do not have this
  property, and it remains open whether a locally finite
  axiomatization is possible in this case.
\end{remark}

\noindent We conclude the section with the announced stand-alone proof
of one-step completeness of the rules for graded modal logic:
\begin{lemma}\label{lem:graded-completeness}
  The rules $(RG1)$, $(A1)$, $(A2)$, and $(RN)$
  (Example~\ref{expl:rules}.\ref{item:rules-graded}) are one-step
  complete for graded modal logic.
\end{lemma}
\begin{proof}
  Let $\tau:V\to\Pow X$. We extend the notion of one-step
  $\tau$-consistency to infinite conjunctive clauses, i.e.\ infinite
  sets of literals, over $\Lambda(V)$ in the evident way: a set $\Phi$
  of literals over $\Lambda(V)$ is \emph{one-step $\tau$-consistent}
  if for all $\rho_1,\dots,\rho_n\in\Phi$,
  $\rho_1\land\dots\land\rho_n$ is one-step $\tau$-consistent. As
  usual, one-step completeness dualizes to showing that a given
  one-step $\tau$-consistent clause $\phi$ is \emph{one-step
    $\tau$-satisfiable}, i.e.\
  $\Sem{\phi}_{\Bag X,\tau}\neq\emptyset$. We can assume that $X$ is
  finite~(see~\cite[Proposition~23]{Schroder07},
  \cite[Lemma~30]{MyersEA09}), and then that $\tau:V\to\Pow X$ is
  surjective. By a standard argument, there exists a maximal infinite
  conjunctive clause $\Phi$ over $\Lambda(V)$ that is one-step
  $\tau$-consistent and contains $\phi$. (Note that the argument is
  entirely generic up to this point.) We construct a multiset
  $\mu\in\Bag X$ that satisfies $\Phi$ over $\tau$, i.e.\
  $\mu\in\bigcap_{\phi\in\Phi}\Sem{\phi}_{TX,\tau}$. Specifically, we
  define $\mu$ as an $\Nat\cup\{\infty\}$-valued measure on $X$ by
  \begin{equation*}
    \mu(A)=\max\{k+1\mid\gldiamond{k}a\in\Phi,\tau(a)=A\}
  \end{equation*}
  for $A\subseteq X$, where by convention $\max\emptyset=0$. To see
  that $\mu$ is really a measure, note first that $\mu(\emptyset)=0$
  by $(RN)$, applied to $a$ such that $\tau(a)=\emptyset$. Moreover,
  $\mu$ is additive, i.e.\ $\mu(A\cup B)=\mu(A)+\mu(B)$ for disjoint
  $A,B$. (Since $X$ is finite, there is no issue about $\mu$ being
  $\sigma$-additive.) Here, $\ge$ is by $(A2)$ and $(RN)$, applied to
  $a,b,c,d$ such that $\tau(a)=A$, $\tau(b)=B$, $\tau(d)=\emptyset$,
  and $\tau(c)=A\cup B$. Similarly, $\le$ is by $(A1)$ applied to
  $a,b,c$ such that $\tau(a)=A$, $\tau(b)=B$, and $\tau(c)=A\cup B$. It
  remains to show that $\mu\in\Sem{\gldiamond{k}a}_{TX,\tau}$ iff
  $\gldiamond{k}a\in\Phi$. `If' is immediate by the definition of
  $\mu$, and `only if' is by $(RG1)$. \rightqed
\end{proof}

\section{Constructive Fixpoints}\label{sec:constructive}

\noindent Our next aim is to prove that the Lindenbaum algebra of
$\Lang_\sharp$ is \emph{constructive}, i.e.\ its fixpoints can be
iteratively approximated in $\omega$ steps. In terms of consistency of
formulas, this means that whenever a formula of the form
$\sharp_{\gamma}\phi \land \psi$ is consistent, then already
$\gamma^i(\phi)(\bot) \land \psi$ is consistent for some $i<\omega$;
this fact will play a pivotal role in the model construction in
Section~\ref{sec:models}.  We begin by introducing the requisite
algebraic tools.

\begin{definition}
  A \emph{$\Lambda$-algebra} is a Boolean algebra $A$ with a monotone
  operation $\hearts^{A}: A \to A$ for each $\hearts \in \Lambda$.  In
  a $\Lambda$-algebra $A$, a modal fixpoint scheme $\phi$ with $n$
  variables is interpreted as an operation $\phi^{A}: A^{n} \to A$ in
  the evident way. A \emph{$\sharp$-algebra} is a $\Lambda$-algebra~$A$
  that is endowed with operations $\sharp_{\gamma}^{A}$ for each
  $\gamma\in\Gamma$ such that for each $a \in A$,
  $\sharp_{\gamma}^{A}(a)$ is the least fixpoint of the map
  $\gamma^{A}(a,-): A \to A$ (in particular, these fixpoints
  \emph{exist} in a $\sharp$-algebra).  In a $\sharp$-algebra $A$,
  every fixpoint formula with variables $\phi(v_1,\dots,v_n)$ is
  interpreted as a function $\phi^A:A^n\to A$. We say that $A$
  \emph{validates} a rule $R = \phi/\psi$ if
  $\psi^{A}(a_{1},\ldots,a_{n}) = \top$ whenever
  $\phi^{A}(a_{1},\ldots,a_{n}) = \top$, for $a_1,\dots,a_n\in A$.  An
  \emph{$\Lang_\sharp$-algebra} is a $\sharp$-algebra $A$ that
  validates every rule in our fixed set $\Rules$ of one-step rules.
\end{definition}
In the tradition of algebraic logic, $\Lang_\sharp$-algebras provide
an algebraic encoding of our proof system. More specifically, we will
be interested in the \emph{Lindenbaum algebra} $A(\Lang_\sharp)$ of
our logic.  As usual, this algebra is defined as the quotient of the
set $\FLang_\sharp$ of formulas under the congruence relation $\equiv$
of provable equivalence ($\phi \equiv \psi$ iff
$\entails\phi\modiff\psi$), equipped with the algebra structure that
just interprets every connective as itself. (The congruence property
of $\equiv$ follows from the fact that both the modalities and the
fixpoint operators come with congruence rules, the former by one-step
completeness and the latter by Lemma~\ref{lem:fp-mon}.) An easy
induction shows that every formula $\phi$ is interpreted as the
element $\phi^{A(\Lang_\sharp)} = [\phi]$ of this algebra, where
$[\phi]$ denotes the equivalence class of $\phi$ under $\equiv$; when
there is no danger of confusion, we may write $\phi$ in place of
$[\phi]$.  The Kozen-Park axiomatization ensures that
$A(\Lang_\sharp)$ actually is an $\Lang_\sharp$-algebra:
  \begin{lemma}
    The Lindenbaum algebra is the initial
    $\Lang_\sharp$-algebra.
  \end{lemma}
  \begin{proof}
    The unfolding axiom makes $\sharp_\gamma\phi$ a prefixpoint of
    $\phi^{A(\Lang_\sharp)}(\phi,-)$ in $A(\Lang_\sharp)$, and the
    fixpoint induction rule ensures that $\sharp_\gamma\phi$ is the
    least such, because every element of $A(\Lang_\sharp)$ is the
    denotation of a formula. This shows that $A(\Lang_\sharp)$ is an
    $\Lang_\sharp$-algebra.  Initiality is then straightforward: if
    $B$ is an $\Lang_\sharp$-algebra, then $B$ validates the unfolding
    axiom and the fixpoint induction rule so that we have a
    well-defined map $f:A(\Lang_\sharp)\to B$ given by
    $f([\phi])=\phi^B$.  By construction of the algebra structure on
    $A(\Lang_\sharp)$, $f$ is homomorphic w.r.t.\ all algebraic
    operators (Boolean, modal, and fixpoint). Uniqueness of $f$ is
    shown by induction over formulas in the usual way, using the fact
    that the fixpoint operators are explicitly included in the algebra
    structure. \rightqed
  \end{proof}
In these terms, our target property is phrased as follows.
\begin{definition}
  We say that $\gamma\in\Gamma$ is \emph{constructive} if
  \begin{equation*}
    \sharp_\gamma\phi = \bigvee_{i<\omega}\gamma(\phi)^i(\bot)
  \end{equation*}
  in the Lindenbaum algebra $A(\Lang_\sharp)$, i.e.\ if
  $\entails\sharp_\gamma\phi\to\psi$ whenever
  $\entails\gamma(\phi)^i(\bot)\to\psi$ for all $i<\omega$. If all
  $\gamma\in\Gamma$ are constructive, then \emph{$A(\Lang_\sharp)$ is
    constructive}.
\end{definition}
\noindent We explicitly state the dual formulation of this property:
\begin{lemma}\label{lem:constructive-dual}
  Let $\gamma$ be constructive. If $\sharp_\gamma\phi\land\psi$ is
  consistent, then $\gamma(\phi)^i(\bot)\land\psi$ is consistent for
  some $i<\omega$.
\end{lemma}
\noindent The central tool for proving constructivity, introduced by
Santocanale~\citeyear{Santocanale08} and featuring prominently in
subsequent work by Santocanale and
Venema~\citeyear{SantocanaleVenema10}, is the notion of a finitary
$\CO$-adjoint:
\begin{definition}
  We say that a formula $\phi$ with a variable $x$ is an
  \emph{$\CO$-adjoint} if for all $\psi\in\FLang_\sharp$, there
  exists a finite set $G_{\phi}(\psi)$ of formulas such that
  for all $\rho\in\FLang_\sharp$,
  \begin{equation*}
  \entails\phi(\rho)\to\psi\text{ iff }\entails\rho\to \chi
  \text{ for some $\chi\in G_{\phi}(\psi)$},    
  \end{equation*}
  i.e.\ $\phi(\rho)\le\psi$ in $A(\Lang_\sharp)$ iff $\rho\le \chi$
  for some $\chi\in G_{\phi}(\psi)$. Moreover, $\phi$ is a
  \emph{finitary $\CO$-adjoint} if $G_{\phi}$ can be chosen
  such that for every $\psi$, the closure of $\psi$ under
  $G_{\phi}$,
%  i.e.\ the smallest set $\CA$ of formulas containing $\psi$ and 
%  such that $\chi\in\CA$ implies $G_{\gamma(\phi)}(\chi)\subseteq\CA$, 
  i.e.\ the smallest set $\CA$ with $\psi\in \CA$ and $\chi\in\CA$
  $\Rightarrow$ $G_{\phi}(\chi)\subseteq\CA$, is finite. We say that a
  modal fixpoint scheme $\gamma\in\Gamma$ is a \emph{(finitary) $\CO$-adjoint}
  if $\gamma(\phi)$ is a (finitary) $\CO$-adjoint for all
  $\phi\in\FLang_\sharp$.
\end{definition}
\begin{lemma}~\cite{Santocanale08} Every finitary $\CO$-adjoint is
  constructive.
\end{lemma}

\begin{definition}\label{def:toplevel}  Let
  $\phi$ be a formula (possibly with variables) such that $\phi$
  contains some modal operator and every occurrence of a variable
  in~$\phi$ is in the scope of some modal operator. Let $\phi'$ be the
  equivalent formula that arises by unfolding all top-level
  occurrences of $\sharp$ once (where an occurrence is
  \emph{top-level} if it is not in the scope of a modal or fixpoint
  operator). By guardedness of fixpoint operators, $\phi'$ is of the
  form $\phi_0\sigma$, where $\phi_0\in\Prop(\Lambda(V))$ is injective
  and $\sigma$ is a substitution.
  % $\phi_0\in\Prop(\Lambda(V_\phi))$ and $\sigma_\phi$ is a
  % substitution; we assume normalized variable names
  % $a_\rho\in V_\phi$ such that $\sigma_\phi(a_\rho)=\rho$. 
  We refer to the equivalence $\phi\equiv\phi_0\sigma$ (or more
  precisely the pair $(\phi_0,\sigma)$) as the \emph{top-level
    decomposition} of $\phi$.
\end{definition}
\noindent The first step in the proof of $\CO$-adjointness for a large
class of operators is a generalization of the rigidity lemma
of~\cite{Santocanale08}:
\begin{lemma}[Rigidity]\label{lem:osc}
  Let $\psi$ be a clause over % $\Lambda(A(\Lang_\sharp))$.
  $\Lambda(\FLang_\sharp)$, with top-level decomposition
  $\psi\equiv\psi_0\sigma_0$. Then $\psi$ is provable iff there exists
  an $\Rules$-derivable monotone
  one-step rule $\phi/\chi$ such that $\phi\sigma_0$ is provable and
  $\chi\PLentails\psi_0$. %in this case, $\sigma$ factors
  % as $\sigma=\sigma_0\circ\tilde\sigma$ where $\tilde\sigma:V\to V$ is
  % a renaming, $\chi\tilde\sigma$ is contracted, and
  % $\chi\tilde\sigma\PLentails\psi_0$
\end{lemma}

\noindent The proof relies on the \emph{one-point extension} of an
algebra (so called because it mimics the addition of a new root point
in a coalgebraic model on the algebraic side), in generalization of a
similar construction by Santocanale and
Venema~\citeyear{SantocanaleVenema10}:

Let $A$ be a countable $\Lang_\sharp$-algebra, let $\Space(A)$ be the
set of ultrafilters of $A$, fix a surjective map $\sigma:V\to A$, and
let an injective conjunctive clause $\rho$ over $\Lambda(V)$ be
one-step $\theta$-consistent for $\theta:V\to\Pow(\Space(A))$ given by
$\theta=j\circ\sigma$ (applicative composition) where
$j:A\to\Pow(\Space(A))$ is the usual canonical map
$j(a)=\{u\in\Space(A)\mid a\in u\}$.  We construct the one-point
extension $A^\rho$, an $\Lang_\sharp$-algebra emulating the addition
of a new point whose successor structure is described by $\rho$, as
follows. To begin, we can find a maximally one-step cutfree
$\theta$-consistent set $\Phi\subseteq\Lambda(V)$ such that
$\Phi\PLentails\rho$.  As we emulate adding a single point, the
carrier of $A^\rho$ is $A \times 2$ where $2$ is the Boolean algebra
$\{\bot,\top\}$; we thus have projection maps $\pi_1:A^\rho\to A$ and
$\pi_2:A^\rho\to 2$. We make $A^\rho$ into a $\Lambda$-algebra by
putting
\begin{equation*}
  \hearts^{A^\rho}(a,d)=(\hearts^A(a),\hearts^\rho(a)),
\end{equation*}
 where $\hearts^\rho:A\to 2$ is defined by
%\begin{equation*}
 $\hearts^\rho(a)=\top$ iff $\hearts a\in\Phi\sigma$.
%\end{equation*}
 (Thus, $\hearts^{A^\rho}(a,d)$ is independent of $d$, in agreement
 with the semantic fact that the interpretation of modal operators
 depends only on the successor structure of the current state, not on
 the current state itself.)  We define a valuation
 $\hat\sigma:V\to A^\rho$ by
 \begin{equation*}
   \hat\sigma(v)=
   \begin{cases}
     (\sigma(v),\bot)& \text{if $v$ occurs in a positive literal in $\rho$}\\
     (\sigma(v),\top)& \text{otherwise.}
   \end{cases}
 \end{equation*}
We then have
 \begin{equation}\label{eq:one-point-sat}
   \rho\hat\sigma>\bot\text{ in }A^\rho:
\end{equation}
by definition of $A^\rho$, the second component of $\rho\hat\sigma$ is
$\top$ because $\Phi\sigma\PLentails\rho\sigma$.
\begin{lemma}\label{lem:one-point-alg}
%\begin{equation}
%\label{eq:one-point-alg}
The algebra $A^\rho$ is an $\Lang_\sharp$-algebra.
%\end{equation}
\end{lemma}
\begin{proof}
  The proof that $A^\rho$ is a $\sharp$-algebra is a simple
  application of Bekic's theorem, as in~\cite{SantocanaleVenema10}:
  since $A^\rho$ has carrier $A\times 2$, the fixpoint definition of
  $(\psi,c)=\sharp_\gamma(\phi,b)$ in $A^\rho$ can be seen as a
  mutually recursive definition of two variables $\psi,c$. By
  guardedness of fixpoints and the interpretation of modalities in
  $A^\rho$, the definition of the first variable $\psi$ does not
  mention the second variable $c$, so we can calculate the solution
  for $\psi$ separately, using the fact that $A$ is a
  $\sharp$-algebra, and then replace $\psi$ with its solution
  $\sharp_\gamma(\psi)$ in the recursive definition of the second
  variable $c$. We end up with $c$ being defined as the least fixpoint
  of a monotone function on $2$, which exists because $2$ is a
  complete lattice. 

  It remains to prove that $A^\rho$ validates the one-step rules in
  $\Rules$.  The first component of $A^\rho$ behaves just like $A$, so
  that we have to verify the rules only on the second component, $2$.
  That is, whenever we have a one-step rule $\chi/\psi\in\Rules$ and a
  valuation $\tau:V\to A^\rho$ such that $\chi\tau=\top$ in $A^\rho$,
  we have to prove that $\pi_2\psi\tau=\top$.  Since $\psi\tau$
  depends only on $\pi_1\tau$, we thus have to prove that whenever we
  have $\tau':V\to A$ such that $\chi\tau'=\top$ in $A$, the
  interpretation of $\psi\tau'$ in $2$ is $\top$, where the
  interpretation of $\psi\tau'$ is determined by means of the
  $\hearts^\rho$ and the Boolean algebra structure of $2$.  Since
  $\sigma$ is surjective, we have $\tau'=\sigma\circ\hat\tau$
  (applicative composition) for some renaming $\hat\tau:V\to V$.
  % Formally, take right inverse kappa:A -> V of sigma
  % and put hat tau = kappa tau, then
  % sigma hat tau = sigma kappa tau = tau.
  It now follows that $\Phi\PLentails\psi\hat\tau$: Otherwise, by
  maximality of of $\Phi$, we would have
  $\Phi\PLentails\neg\psi\hat\tau$. Since $\Phi$ is one-step
  $\theta$-consistent, this implies that
  $\Sem{\neg\chi\hat\tau}_{\Space(A),\theta}\neq\emptyset$. But since
  $\theta\circ\hat\tau=j\circ\sigma\circ\hat\tau=j\circ\tau'$,
  $\Sem{\neg\chi\hat\tau}_{\Space(A),\theta}$ is the image of
  $\neg\chi\tau'\in A$ under the Boolean homomorphism
  $j:A\to\Pow(\Space(A))$. This implies $\neg\chi\tau'\neq\bot$, in
  contradiction to $\chi\tau'=\top$. \rightqed
\end{proof}

\noindent
In consequence of the fact that $A(\Lang_\sharp)$ is the
\emph{initial} $\Lang_\sharp$-algebra, we thus have
\begin{lemma}\label{lem:one-point}
  Let $\sigma:V\to A(\Lang_\sharp)$ be surjective, and let $\rho$ be a
  conjunctive clause over $\Lambda(V)$. If $\rho$ is one-step
  $\theta$-consistent for
  $\theta(v)=\{u\in\Space(A(\Lang_\sharp))\mid\sigma(v)\in u\}$, then
  $\rho\sigma$ is consistent, i.e.\ $\rho\sigma>\bot$ in
  $A(\Lang_\sharp)$.
\end{lemma}
\begin{proof}
  Let $f$ be the unique $\Lang_\sharp$-algebra homomorphism
  $A(\Lang_\sharp)\to A(\Lang_\sharp)^\rho$, and take
  $\hat\sigma:V\to A(\Lang_\sharp)^\rho$ as above. Note that the first
  projection $\pi_1:A(\Lang_\sharp)^\rho\to A(\Lang_\sharp)$ is a
  homomorphism of $\Lang_\sharp$-algebras by construction of
  $A(\Lang_\sharp)^\rho$. By the uniqueness part of initiality, it
  follows that $\pi_1\circ f=\id$. For $a\in V$, we therefore have
  $\pi_1(f(\sigma(a))=\sigma(a)=\pi_1(\hat\sigma(a))$; moreover, if
  $a$ occurs in a positive literal in $\rho$, then
  $\pi_2(f(\sigma(a)))\ge\bot=\pi_2(\hat\sigma(a))$, and if $a$ occurs
  in a negative literal in $\rho$ (hence does not occur in a positive
  literal, since $\rho$ is injective), then
  $\pi_2(f(\sigma(a)))\le\top=\pi_2(\hat\sigma(a))$. Since the order
  on $A(\Lang_\sharp)^\rho$ is componentwise, this implies that
  $f\circ\sigma(a)\ge\hat\sigma(a)$ if $a$ occurs in a positive
  literal in $\rho$, and $f\circ\sigma(a)\le\hat\sigma(a)$ if $a$
  occurs in a negative literal in $\rho$. Since $f$ is homomorphic and
  the modalities are monotone, it follows by~\eqref{eq:one-point-sat}
  that $f(\rho\sigma)=\rho(f\circ\sigma)\ge\rho\hat\sigma>\bot$. Since
  $f$ preserves $\bot$, this implies $\rho\sigma>\bot$. \rightqed
\end{proof}
\noindent From Lemma~\ref{lem:one-point}, one easily proves rigidity
(Lemma~\ref{lem:osc}) using the fact that every consistent formula is
contained in some ultrafilter of $A(\Lang_\sharp)$:

\begin{proof}[(Lemma~\ref{lem:osc})]
  `If' is clear. For `only if', we prove the dual statement: Let
  $\psi$ be a conjunctive clause over $\Lambda(\FLang_\sharp)$, with
  top-level decomposition $\psi\equiv\psi_0\sigma_0$, such that
  $\neg\phi\sigma_0$ is consistent for all $\Rules$-derivable monotone
  one-step rules $\phi/\chi$ and all substitutions $\sigma$ such that
  $\chi\PLentails\neg\psi_0$; we show that $\psi$ is consistent.
  W.l.o.g.\ $\sigma_0$ is a surjection $V\to\FLang_\sharp$, which
  we prolong to a surjection $\bar\sigma_0:V\to A(\Lang_\sharp)$.
  By Lemma~\ref{lem:one-point}, it suffices to prove that $\psi_0$ (a
  conjunctive clause over $\Lambda(V)$) is one-step
  $\theta$-consistent for
  $\theta(v)= \{u\in\Space(A(\Lang_\sharp))\mid\bar\sigma_0(v)\in
  u\}$.
  Thus, let $\phi/\chi$ be an $\Rules$-derivable rule such that
  $\chi\PLentails\neg\psi_0$.  We have to show that
  $\Space(A(\Lang_\sharp)),\theta\not\models\phi$.  By assumption and
  because $\chi\PLentails\neg\psi_0$, we have that $\neg\phi\sigma_0$
  is consistent, hence $\neg\phi\bar\sigma_0>\bot$ in
  $A(\Lang_\sharp)$, so that there exists
  $u\in\Space(A(\Lang_\sharp))$ with $\neg\phi\bar\sigma_0\in u$.  Now
  one shows by induction over $\rho\in\Prop(V)$ that
  $u\in\Sem{\rho}_{\Space(A(\Lang_\sharp)),\theta}$ iff
  $\rho\bar\sigma_0\in u$: the cases for Boolean connectives are by
  the ultrafilter property of $u$, and the base case $\rho=a\in V$ is
  by definition of~$\theta$. In particular, we have
  $u\in\Sem{\neg\phi}_{\Space(A(\Lang_\sharp)),\theta}$, showing that
  $\Space(A(\Lang_\sharp)),\theta\not\models\phi$ as required.
  \rightqed
\end{proof}
In a nutshell, rigidity enables us to prove $\CO$-adjointness of all
(monotone) modal operators, and even more generally all modal fixpoint schemes
in which the recursion variable $x$ occurs at uniform depth (such as
$\Box\Diamond x\land\Diamond\Box x$). Formally:

\begin{definition}\label{def:uniform}
  A formula $\phi$ with variables is \emph{uniform of depth $k$} if
  every occurrence of the fixed recursion variable $x$ in $\phi$ is
  in the scope of exactly $k$ modal operators (including the case that
  $x$ does not occur in $\phi$; recall moreover that variables
  never occur under fixpoint operators). Moreover, $\phi$ is
  \emph{uniform} if $\phi$ is uniform of depth $k$ for some $k$;
  the minimal such $k$ is the \emph{depth of uniformity} of $\phi$.
  Formally, uniform formulas are defined
  inductively by following inductive clauses.
  \begin{itemize}
  \item $x$ is uniform of depth $0$;
  \item formulas not containing $x$ are uniform of depth $k$, for any $k$;
  \item any Boolean combination of uniform formulas of depth $k$ is
    uniform of depth $k$; and
  \item if $\psi$ is uniform of depth $k$, then $\hearts\psi$ is
    uniform of depth $k+1$, for $\hearts\in\Lambda$.
  \end{itemize}
  % \end{definition}

% f For a variable $v\in V$, we inductively define
%   the sets $\UD_k(v)$ for $k \in \omega$:
% \begin{itemize}
%  \item $v\in \UD_0(v)$, and $\psi\in\UD_k(v)$ if $\psi$ does not contain $v$;
%  \item each set $\UD_k(v)$ is closed under Boolean combinations;
%  \item if $\psi\in\UD_k(v)$, then $\hearts\psi\in\UD_{k+1}(v)$, for all $\hearts\in\Lambda$.
%  \end{itemize}
%If $\psi\in\UD_k(v)$ we say that $\psi$ is \emph{uniform of depth $k$ in $v$},
%and if $\psi \in\UD_k(v)$ for some $k$ we say that $\psi$ is \emph{uniform in $v$}.
\end{definition}

%\begin{definition}\label{def:uniform}
%  Let $v\in V$. We say that a formula $\psi$ with variables is
%  \emph{uniform of depth $k$ in $v$} if all occurrences of $v$ are
%  under the scope of exactly $k$ modal operators but not under the
%  scope of any fixpoint operators, and \emph{uniform in $v$} if
%  $\psi$ is uniform of depth $k$ in $v$ for some $k$, which we then
%  call the \emph{depth of $v$-uniformity} of $\phi$ (noting that
%  $k$ is uniquely determined unless $\psi$ does not contain any
%  occurrence of $v$). 
\noindent Finitariness of $\CO$-adjoints will use the standard
Fischer-Ladner closure:
\begin{definition}
  A set $\Sigma$ of formulas is \emph{Fischer-Ladner closed} if
  $\Sigma$ is closed under subformulas and negation (where the
  negation of a negated formula $\neg\phi$ is taken to be~$\phi$), and
  whenever $\sharp_\gamma\phi\in\Sigma$, then
  $\gamma(\phi,\sharp_\gamma\phi)\in\Sigma$.  We denote the
  \emph{Fischer-Ladner closure} of a formula $\phi$, i.e.\ the
  smallest Fischer-Ladner closed set containing $\phi$, by
  $\FL(\phi)$.
\end{definition}
\begin{lemma}~\cite{Kozen83} \label{lem:fl} The set $\FL(\phi)$ is
  finite. % Section~\ref{sec:prelim})
\end{lemma}

\noindent The further development revolves largely around
derivable rules:
\begin{definition}
  A \emph{rule} $R=\phi/\psi$ consists of a \emph{premise} $\phi$ and
  a \emph{conclusion} $\psi$, both being formulas with variables. The
  rule $R$ is \emph{derivable} if $\psi$ can be derived from the
  assumption $\phi$ using the rules of the system (propositional
  reasoning, unfolding, fixpoint induction, modal rules).
\end{definition}
\begin{lemma}\label{lem:positive-monotone}
  If $\phi$ is positive in $x$ and $\entails\chi\to\rho$, then
  $\entails\phi(\chi)\to\phi(\rho)$.
\end{lemma}
\begin{proof}
  Induction over $\phi$, simultaneously with a corresponding statement
  on formulas that are negative in $x$. The Boolean cases are
  straightforward, and the cases for fixpoint operators are trivial
  because $x$ never appears under fixpoint operators. The case for a
  modal operator $\hearts\in\Lambda$ is discharged by the fact that by
  monotonicity of $\hearts$, the monotonicity rule
  $a\to b/\hearts a\to\hearts b$ is one-step sound and therefore, by
  Lemma~\ref{lem:oss-derive}, $\Rules$-derivable. \rightqed
\end{proof}

\begin{lemma}\label{lem:uniform}
  Let $\psi=\psi(x)$ be a uniform formula, and put
  \begin{equation*}
    G=\{\phi\in\Prop(\FL(\psi))\mid\phi/\psi\text{ derivable},
    \phi\text{ uniform of depth $0$}\}.
  \end{equation*}
  Then given a formula $\rho$, $\psi(\rho)$ is provable iff
  $\phi(\rho)$ is provable for some~$\phi(x)\in G$.
\end{lemma}
\begin{proof}
  `If' holds by definition of admissibility. We prove `only if' by
  induction over the depth of uniformity, with trivial base
  case. Thus, let $\psi$ be uniform of depth $k>0$, and let
  $\psi(\rho)$ be provable. By unfolding (guarded) top-level fixpoints
  (i.e.\ those not in the scope of a modality) and then applying
  propositional reasoning to transform into CNF, we reduce to the case
  that $\psi$ is a clause over $\Lambda(\FL(\psi))$; note here that
  these transformations remain within $\Prop(\FL(\psi))$.  Let
  $\psi\equiv\psi_0\sigma_0$ be the top-level decomposition of $\psi$;
  then the top-level decomposition of $\psi(\rho)$ has the form
  $\psi_0\sigma_0^\rho$ where $\sigma_0^\rho(a)=\sigma_0(a)(\rho)$. By
  rigidity (Lemma~\ref{lem:osc}), we have an $\Rules$-derivable
  monotone one-step rule $\chi/\psi'$ such that $\chi\sigma_0^\rho$ is
  provable and $\psi'\PLentails\psi_0$.  Then the rule
  $\chi\sigma_0/\psi$ is derivable, $\chi\sigma_0\in\Prop(\FL(\psi))$
  is uniform of depth $k-1$, and $(\chi\sigma_0)(\rho)$ is provable
  (being equal to $\chi\sigma_0^\rho$).  By the inductive assumption,
  applied to $\chi\sigma_0$, there exists
  $\phi(x)\in\Prop(\FL(\chi\sigma_0))\subseteq\Prop(\FL(\psi))$ such
  that $\phi/\chi\sigma_0$ is derivable, $\phi$ is uniform of degree
  $0$, and $\phi(\rho)$ is provable. Since $\phi/\chi\sigma_0$ and
  $\chi\sigma_0/\psi$ are derivable, so is $\phi/\psi$, hence
  $\phi\in G$, which proves the claim. \rightqed
\end{proof}
 % \begin{proof}[sketch]
 %   Induction over the depth of uniformity, with trivial base case,
 %   using rigidity (Lemma~\ref{lem:osc}) in the inductive step. \qed
 % \end{proof}
\noindent We recall a trick from propositional logic:
\begin{lemma}\label{lem:shannon}
  Let $\phi$ be a formula with variables containing only top-level
  occurrences (Definition~\ref{def:toplevel}) of the variable
  $y$. Then there is a formula $a$ such that
  \begin{equation*}
    \phi[\top/y]\lor\phi[\bot/y]\PLentails\phi[a/y].
  \end{equation*}
\end{lemma}
\begin{proof}
  By standard Boolean expansion. Specifically, $a=\phi[\top/y]$ does
  the job. This is seen by case distinction over $\phi[\top/y]$:
  First, assume $\phi[\top/y]$; then $\phi[\top/y]$ is equivalent to
  $\top$, so $\phi[\top/y]$ propositionally entails
  $\phi[\phi[\top/y]/y]$ because $y$ has only top-level occurrences in
  $\phi$. Second, assume $\neg\phi[\top/y]$. From
  $\phi[\top/y]\lor\phi[\bot/y]$ we then conclude
  $\phi[\bot/y]$. Since by assumption, $\phi[\top/y]$ is equivalent to
  $\bot$, this, again, propositionally entails
  $\phi[\phi[\top/y]/y]$. \rightqed
\end{proof}
% \noindent As an immediate consequence, we have
% \begin{lemma}[Variable elimination]\label{lem:elim}
%   Let $\phi/\psi$ be an derivable rule, where $\psi$ does not contain
%   the variable $y$ and $y$ has only top-level occurrences in
%   $\phi$. Then the rule
%   \begin{equation*}
%     (\phi[\top/y]\lor\phi[\bot/y])/\psi
%   \end{equation*}
%   is derivable. \rightqed
% \end{lemma}
\noindent We are now set to prove the main result of this section:
\begin{theorem}[Finitary $\CO$-adjointness] \label{thm:o-adj}
  If the formula $\psi$ with recursion variable $x$ is positive and
  uniform in $x$, then the operation
  $\psi^{A(\Lang_\sharp)}:A(\Lang_\sharp)\to A(\Lang_\sharp)$ induced
  by $\psi$ is a finitary $\CO$-adjoint.
\end{theorem}
% \begin{proof}[sketch]
%   For $\phi\in\FLang_\sharp$, we have to construct a set
%   $G_\psi(\phi)$ of formulas such that for all
%   $\rho\in\FLang_\sharp$, $\psi(\rho)\entails\phi$ iff
%   $\rho\entails\chi$ for some $\chi\in G_\psi(\phi)$. Now
%   $\psi':\equiv\psi\modimpl\phi$ is uniform. Let
%   $G\subseteq\Pos(FL(\psi'))$ be as in Lemma~\ref{lem:uniform},
%   applied to $\psi'$, and let $G_0$ be a finite set of representatives
%   of $G$ modulo propositional equivalence. Then we can put
%   \begin{equation*}
%     G_\psi(\phi)=\{\chi(\top)\mid\chi\in G_0,
%     \entails\chi(\top)\lor\chi(\bot)\}.\vspace{-2.3em}
%   \end{equation*}
% \qed
%  \end{proof}
\begin{proof}
  For readability, we phrase the arguments using formulas in
  $\FLang_\sharp$ rather than elements of $A(\Lang_\sharp)$. Let
  $\phi\in\FLang_\sharp$. We have to construct a set $G_\psi(\phi)$ of
  formulas such that for all $\rho\in\FLang_\sharp$,
  $\entails\psi(\rho)\modimpl\phi$ iff $\entails\rho\modimpl\chi$ for some
  $\chi\in G_\psi(\phi)$; moreover we require $G_\psi(\phi)$ to be
  finite up to provable equivalence.
  
  The formula $\psi':\equiv\psi\modimpl\phi$ is uniform (as $\phi$
  does not contain $x$). Let $G\subseteq\Prop(\FL(\psi'))$ be as in
  Lemma~\ref{lem:uniform}, applied to $\psi'$; notice that
  $\Prop(\FL(\psi'))$ is finite up to provable equivalence. Then put
  \begin{equation*}
    G_\psi(\phi)=\{\chi(\top)\mid\chi(x)\in G,\entails\chi(\top)\lor\chi(\bot)\}.
  \end{equation*}
  Now let $\entails\rho\modimpl\chi(\top)$ for some $\chi(x)\in G$ such
  that $\entails\chi(\top)\lor\chi(\bot)$. To show that
  $\entails\psi(\rho)\modimpl\phi$, it suffices by construction of $G$ and
  positivity of $\psi$ to prove that
  $\entails\chi(a)\land(\rho\modimpl a)$ for some formula $a$: then it
  follows that $\entails(\psi(a)\modimpl\phi)\land(\rho\modimpl a)$,
  and hence by Lemma~\ref{lem:positive-monotone} that
  $\entails\psi(\rho)\modimpl\phi$. Since $\chi\land(\rho\modimpl x)$
  is uniform of depth $0$, existence of such an $a$ follows by
  Lemma~\ref{lem:shannon} once we show that
\begin{equation*}
  \entails(\chi(\top)\land(\rho\modimpl\top))\lor
  (\chi(\bot)\land(\rho\modimpl \bot)),
\end{equation*}
equivalently 
\begin{equation*}
  \entails\chi(\top)\lor(\neg\rho\land\chi(\bot)).
\end{equation*}
By distributing disjunction over conjunction, this last formula is
equivalent to
\begin{equation*}
(\rho\modimpl\chi(\top))\land(\chi(\bot)\lor\chi(\top))
\end{equation*}
and hence provable by assumption.

Conversely, let $\entails\psi(\rho)\modimpl\phi$. By
Lemma~\ref{lem:uniform}, there exists $\chi\in G$ such that
$\chi(\rho)$ is provable. % By positivity of
% $\psi$ and Lemma~\ref{lem:positive-monotone}, the rule
% \begin{equation*}
% \infrule{(x\modimpl y)\land\chi(y)}{\psi'}
% \end{equation*}
% is derivable. Since $\chi$ is uniform of depth $0$, eliminating $y$
% from the premise of this rule according to Lemma~\ref{lem:elim} yields
% an derivable rule
% \begin{equation}\label{eq:elim-rule-chi}
%   \infrule{(\neg x\land\chi(\bot))\lor\chi(\top)}{\psi'},
% \end{equation}
% so the formula
% \begin{equation*}
% \chi'(x):\equiv(x\modimpl\chi(\top))\land(\chi(\bot)\lor\chi(\top)),
% \end{equation*}
% being propositionally equivalent to the premise
% of~\eqref{eq:elim-rule-chi}, also belongs to $G$.
Since $\chi$ is uniform of depth $0$, $\chi(\rho)$ is propositionally
equivalent to its Boolean expansion
$(\rho\modimpl\chi(\top))\land(\neg\rho\modimpl\chi(\bot))$ and hence
propositionally entails $\chi(\bot)\lor\chi(\top)$, which is therefore
provable. That is, we have $\chi(\top) \in G_\psi(\phi)$ and
$\entails\rho\modimpl\chi(\top)$, as required.

This proves that $\psi^{A(\Lang_\sharp)}$ is $\CO$-adjoint. From the
above description of $G_\psi$, one sees immediately that
$\psi^{A(\Lang_\sharp)}$ is in fact a finitary $\CO$-adjoint, as
$\Prop(\FL(\psi'))$ is closed under~$G_\psi$ and finite up to provable
equivalence.\rightqed
\end{proof}
\noindent Using uniform formulas as a base, we can now exploit some
known closure properties of finitary
$\CO$-adjoints~\cite{Santocanale08}.
\begin{definition}\label{def:admissible}
  The set of \emph{admissible} modal fixpoint schemes is the closure
  of the set of uniform modal fixpoint schemes under disjunction,
  conjunction with modal fixpoint schemes not containing the recursion
  variable $x$, and substitution for $x$, the latter in the sense that
  if $\gamma(x)$ and $\delta$ are admissible, then $\gamma(\delta)$ is
  admissible.
  % A fixpoint
  % operator $\sharp_\gamma$ is admissible if $\gamma$ is admissible.
\end{definition}
\begin{corollary}\label{cor:admissible}
  If $\gamma\in\Gamma$ is admissible, then $\gamma$ is a finitary
  $\CO$-adjoint, and hence constructive.
\end{corollary}
\begin{proof}
  The set of finitary $\CO$-adjoints is closed under joins, meets with
  constants, and composition~\cite{Santocanale08}. \rightqed
\end{proof}
\noindent From now on, \emph{we require that every $\gamma\in\Gamma$
  is admissible}, and hence $A(\Lang_\sharp)$ is constructive; a flat
coalgebraic fixpoint logic satisfying this requirement will be called
\emph{admissible}. All fixpoint operators mentioned in
Example~\ref{expl:logics} are based on admissible fixpoint schemes (in
fact, on uniform ones).

\begin{remark}\label{rem:sv}
  The sufficient criterion for $\CO$-adjointness given by Santocanale
  and Venema~\citeyear{SantocanaleVenema10} is that modal fixpoint
  schemes be \emph{harmless}, which in the single-modality case means
  that modal fixpoint schemes $\gamma,\delta$ are generated by the
  grammar
  \begin{equation*}\textstyle
    \gamma,\delta ::=\top\mid x\mid \gamma\lor\delta\mid
    \chi\land\gamma\mid \Land_{i=1}^k\Diamond\phi_i\mid\Box\phi
  \end{equation*}
  where $\chi$ is a modal formula not mentioning the recursion
  variable $x$ (but possibly mentioning argument variables).  This
  notion is incomparable to admissibility in the sense of
  Definition~\ref{def:admissible}; e.g.\
  $\Diamond x \land\Diamond \Diamond x$ is harmless but not
  admissible, and $\Diamond \Box x \land\Box\Diamond x$ is admissible
  (in fact, uniform) but not harmless. We leave it as an open problem
  to find a sufficient criterion for $\CO$-adjointness that subsumes
  both admissibility and harmlessness. The fixpoint schemes generating
  the CTL operators are both harmless and admissible.
\end{remark}

\section{The Model Construction}\label{sec:models}

\noindent We proceed to describe a model construction that uses sets
of \emph{timed-out formulas} as states; a timed-out formula has some
of the fixpoints that appear in it annotated with time-outs indicating
how often they need to be unfolded. Our time-outs are related to
Kozen's $\mu$-counters~\cite{Kozen83} but, as indicated, are
integrated into formulas rather than maintained independently in a
tableau construction. The use of time-outs is justified by
constructivity of fixpoint operators as proved in the previous
section. 

Since only some of the fixpoint subformulas contained in a state will
be annotated with finite time-outs, only one implication of the truth
lemma will hold (every state satisfies the timed-out formulas it
contains but not conversely -- the model we construct will be finite,
so every fixpoint will be satisfied with some time-out, which the
state may fail to specify). Consequently, in the inductive proof of
the truth lemma, the step for negation would fail. We therefore work
with formulas in negation normal form, defined in detail as follows. A
modal fixpoint scheme is in \emph{negation normal form (NNF)} if it
can be generated by the grammar
\begin{equation*}
  \gamma,\delta::=\bot\mid\top\mid v\mid \gamma\land\delta\mid
  \gamma\lor\delta\mid \hearts\gamma\mid\dual\hearts\gamma,
\end{equation*}
(recall that $\dual\hearts$ abbreviates $\neg\hearts\neg$). We can
clearly transform every modal fixpoint scheme into a provably
equivalent one in NNF (recall that modal fixpoint schemes are positive
in all variables), and therefore assume from now on that $\Gamma$
consists of modal fixpoint schemes in NNF (this does not substantially
affect the syntax as the modal fixpoint schemes just serve as indices
of fixpoint operators). A formula $\phi$ is in \emph{NNF} if it can be
generated by the grammar
\begin{equation*}
  \phi,\psi::=\bot\mid\top\mid \phi\land\psi\mid \phi\lor\psi\mid
  \hearts\phi\mid\dual\hearts\phi\mid
  \sharp_\gamma\phi\mid
  \flat_{\gamma}\phi.
\end{equation*}
The dual $\dual\gamma$ of a modal fixpoint scheme $\gamma$ as defined
in Section~\ref{sec:prelim} is clearly equivalent to the modal
fixpoint scheme obtained from $\gamma$ by swapping $\land$ with
$\lor$, $\top$ with $\bot$, and $\hearts$ with $\dual\hearts$, and we
regard $\dual\gamma$ as being syntactically defined in this way from
now on; e.g.\ if $\Diamond=\dualBox$ then the dual of
$\gamma=p\lor\Box x$ is $\dual\gamma=p\land\Diamond x$. To show that
we can transform every formula into NNF, it suffices as usual to show
that we can implement negation on NNFs; this is by the standard
procedure of pushing negation inside through the other connectives
using provable equivalences. In particular, the equivalence
\begin{equation*}
  \neg\sharp_\gamma\phi\modiff\flat_{\dual\gamma}\neg\phi
\end{equation*}
is trivially provable, as the right hand side is just an abbreviation for
$\neg\sharp_\gamma\neg\neg\phi$. Summing up, we have
\begin{lemma}\label{lem:NNF}
  Every formula is provably equivalent to a formula in NNF.
\end{lemma}
\noindent It remains to adapt the notion of Fischer-Ladner closure; as
no confusion is likely and the changes are rather inessential, we use
the same term:
\begin{definition}
  A set $\Sigma$ of formulas in NNF is \emph{Fischer-Ladner closed} if
  $\Sigma$ is closed under subformulas and fixpoint unfolding, i.e.\
  whenever $\star_\gamma\phi\in\Sigma$ for $\star\in\{\sharp,\flat\}$
  then $\gamma(\phi,\star_\gamma\phi)\in\Sigma$.
\end{definition}
The analogue of Lemma~\ref{lem:fl} remains true for the modified
definition; so from now on we fix a finite Fischer-Ladner closed set
$\Sigma$ of formulas in NNF. We proceed to introduce the announced
notion of time-out:

\begin{definition}
  The set of \emph{timed-out formulas} $\phi,\psi$ is generated by the
  grammar
  \begin{equation*}
    \phi,\psi::=\bot\mid\top\mid\phi\land\psi\mid\phi\lor\psi\mid
    \hearts\phi\mid \dual\hearts\phi\mid
    \sharp^\kappa_\gamma\rho\mid \flat_{\dual\gamma}\rho\qquad
    (\kappa\in\omega+1,\rho\in\Lang_\sharp)
  \end{equation*}
  where $\gamma\in\Gamma,\hearts\in\Lambda$, subject to the
  restriction that $\phi$ is a timed-out formula only in case~$\phi$
  has at most one subformula of the form $\sharp^\kappa_\gamma\chi$
  with $\kappa<\omega$ (which however may occur any number of times),
  and for this $\sharp^\kappa_\gamma\chi$,
  \begin{itemize}
  \item $\sharp^\omega_\gamma\chi$ is not a subformula of $\phi$;
    and 
  \item  whenever $\sharp^\omega_\delta\rho$ is a subformula of
    $\phi$, then $\sharp_\delta\rho$ is a subformula of $\chi$.
  \end{itemize}
  In this case, we define the \emph{time-out} $\tau(\phi)$ of
  $\phi$ to be $\kappa$, and $\tau(\phi)=\omega$ otherwise
  (i.e.\ if $\phi$ does not contain any subformula of the form
  $\sharp^\kappa_\gamma\chi$ with $\kappa<\omega$).  The time-out
  gives the number of steps left until satisfaction of the eventuality
  $\sharp_\gamma\chi$, with time-out $\omega$ signifying an
  unspecified number of steps (note that time-outs are never
  associated with $\flat$-formulas).

  % An \emph{extended timed-out
  %   formula} is a Boolean combination of timed-out formulas.
  We define two translations $s$ and $t$ of timed-out formulas into
  $\Lang_\sharp$, given by commutation with Boolean and modal
  operators,
  $(\flat_{\dual\gamma}\rho)^s=(\flat_{\dual\gamma}\rho)^t=\flat_{\dual\gamma}\rho$,
  and
\begin{equation*}
    (\sharp^\omega_\gamma\rho)^s  =\sharp_\gamma\rho\qquad
    (\sharp^i_\gamma\rho)^s  =
    \gamma(\rho)^i(\bot)\quad(i<\omega) \qquad
    (\sharp^\kappa_\gamma\rho)^t  = \sharp_\gamma\rho.
  \end{equation*}
  Thus, $s$ unfolds fixpoints as prescribed by their time-outs, and
  $t$ just removes time-outs.  Both translations extend to sets of
  formulas. For timed-out formulas $\phi,\psi$, we put
  $\phi\preceq\psi$ iff $\phi^t=\psi^t$ and
  $\tau(\phi)\le\tau(\psi)$. That is, $\phi\preceq\psi$ iff
  $\phi$ is the same as $\psi$ up to possible decrease of the
  time-out. Given a set $\Sigma$ of formulas, a timed-out formula
  $\phi$ is a \emph{timed-out $\Sigma$-formula} if
  $\phi^t\in\Sigma$. 
\end{definition}

\noindent Notice that unlike formulas in NNF, timed-out formulas are
not closed under negation because time-outs can only appear on least
fixpoints. The point of the definition of timed-out formulas is that
every standard formula $\phi$ has at most one candidate subformula at
which one can insert a time-out, namely the greatest element under the
subformula ordering among the subformulas of $\phi$ which are
$\sharp$-formulas, if such a greatest element exists and is not in
scope of a $\flat$-operator. This enables the simple definition of
$\preceq$, which by the preceding discussion has the following
property.
\begin{lemma}\label{lem:linear}
  For every formula $\phi$, the preimage of $\phi$ under the
  translation $t$ is well-ordered by $\preceq$.
\end{lemma}

\noindent At the same time, timed-out formulas are stable under
unfolding:

\begin{lemma}\label{lem:FL-timed-out}
  If $\sharp^\kappa_\gamma\phi$ is a timed-out formula, then so is
  $\gamma(\phi,\sharp^\kappa_\gamma\phi)$.
\end{lemma}
\begin{proof}
    As $\phi$ is a (standard) formula, it is clear that
  $\gamma(\phi,\sharp^\kappa_\gamma\phi)$ cannot contain formulas of the
  form $\sharp^\lambda_\delta\rho$ with $\lambda<\omega$ other than
  $\sharp^\kappa_\gamma\phi$. By well-foundedness of the subformula
  relation, $\gamma(\phi,\sharp^\kappa_\gamma\phi)$ cannot contain
  $\sharp^\omega_\gamma\phi$. Finally, the only way
  $\sharp$-subformulas can arise in
  $\gamma(\phi,\sharp^\kappa_\gamma\phi)$ is as subformulas of
  $\sharp^\kappa_\gamma\phi$. \rightqed
\end{proof}
\noindent We have the expected relationship between the translations
regarding provable entailment:
\begin{lemma}\label{lem:trans-entail}
  \begin{enumerate}
  \item\label{item:trans-entail-t} For every timed-out formula $\phi$,
    $\entails\phi^s\to\phi^t$.
  \item\label{item:trans-entail-s} If $\phi\preceq\psi$ for timed-out formulas $\phi,\psi$, then
    $\entails\phi^s\to\psi^s$.
  \end{enumerate}
 \end{lemma}
\begin{proof}
  By iterated application of the unfolding axiom and monotonicity of
  modal fixpoint schemes. \rightqed
\end{proof}
\noindent States of the model will be sets of formulas satisfying a
timed-out version of the usual expandedness requirement.
\begin{definition}\label{def:timed-out-atom}
  A \emph{timed-out $\Sigma$-atom} is a maximal set $A$ of timed-out
  $\Sigma$-formulas such that
%  \begin{myitemize}
%  \item $\sharp_\gamma\phi^\kappa\in A$ iff
%    $\gamma(\sharp_{\gamma}\phi^{\kappa-1})\in A$.
%\item 
  \begin{itemize}
  \item the translation $t$ is injective on $A$, and
  \item $A^s$ is consistent.
  \end{itemize}
%\end{myitemize}
Here, maximality is w.r.t.\ $\sqsubseteq$ where $A\sqsubseteq B$ iff
for all $\phi\in A$, there exists a (necessarily unique)
$\phi'\in B$ such that $\phi'\preceq\phi$; intuitively: $B$
contains $A$ up to possible decrease of time-outs.
We write $\upA$ for the upwards closure of $A$ under $\preceq$
(i.e.\ if $\phi\in \upA$ and $\phi\preceq\phi'$ then
$\phi'\in \upA$).
\end{definition}
To prove finiteness of the model we construct, we use the fact that
finite product orderings $(\omega+1)^k$ are well-quasi-orders, and in
particular have only finite anti-chains~\cite{lave:well76}.
\begin{lemma}[Timed-out Lindenbaum lemma]\label{lem:to-at-fin}
  \begin{enumerate}
  \item\label{claim:at-fin} The set of timed-out $\Sigma$-atoms is
    finite.
  \item\label{claim:at-lind-timed} For every set $A_0$ of timed-out
    $\Sigma$-formulas such that $A_0^s$ is consistent, there exists a
    timed-out $\Sigma$-atom $A$ such that $A_0\sqsubseteq A$.
  \item\label{claim:at-lind-normal} For every consistent subset
    $C\subseteq\Sigma$, there exists a timed-out $\Sigma$-atom $A$
    such that $C\subseteq A^t$.
  \end{enumerate}
\end{lemma}
\begin{proof}
  Claim~(\ref{claim:at-lind-normal}) follows
  from~(\ref{claim:at-lind-timed}) by converting $C$ into a set of of
  timed-out formulas with time-out $\omega$. We
  prove~(\ref{claim:at-fin}) and~(\ref{claim:at-lind-timed}) in one
  go. Note first that we can assume w.l.o.g.\ that $t$ is injective on
  $A_0$. We define a \emph{$\Sigma$-atom} to be a maximally consistent
  subset of $\Sigma$. The set of $\Sigma$-atoms is finite, so for
  every set $A_0$ as in the statement there exists a $\Sigma$-atom $C$
  such that $A_0^t\subseteq C$. Let $\FA$ be the set of sets $A$ of
  timed-out $\Sigma$-formulas such that $A^s$ is consistent, $t$ is
  injective on $A$, and $A^t=C$; the timed-out $\Sigma$-atoms whose
  $t$-image is $C$ are the $\sqsubseteq$-maximal elements
  of~$\FA$. Now every element $A\in\FA$ is uniquely determined by the
  $C$-tuple $\iota(A)$ of time-outs it induces (explicitly, for
  $\phi\in C$, the $\phi$-component of $\iota(A)$ is $\tau(\phi')$
  where $\phi'$ is the unique formula in $A$ such that
  $(\phi')^t=\phi$), and for $A,B\in\FA$, $A\sqsubseteq B$ iff
  $\iota(B)\le\iota(A)$ in the componentwise ordering. Thus, $\FA$ is
  order-isomorphic to a subset of the finite power $(\omega+1)^C$ of
  the well-ordering $\omega+1$.  It follows from the theory of
  \emph{well-quasi-orders} that $(\omega+1)^C$ is a well-quasi-order,
  which means that for every subset $F$ of $(\omega+1)^C$, every
  element of $F$ is above one of finitely many minimal elements of
  $F$~\cite{lave:well76}; applying this to $\FA$
  proves~(\ref{claim:at-fin}) and~(\ref{claim:at-lind-timed})
  (for~(\ref{claim:at-fin}), recall additionally that the set of
  $\Sigma$-atoms is finite). \rightqed
\end{proof}

\noindent As usual, the proof of the truth lemma will depend on a set
of Hintikka-like properties:

\begin{lemma}\label{lem:pred-at}
  If $A$ is a timed-out $\Sigma$-atom, then 
  \begin{myenumerate}
  \item\label{item:and} if $\phi\land\psi\in A$ then $\phi\in \upA$ and
    $\psi\in \upA$;
  \item\label{item:or} if $\phi\lor\psi\in A$ then $\phi\in \upA$ or
    $\psi\in \upA$;
  \item\label{item:bot} $\bot\notin \upA$;
  \item\label{item:sharp1} if $\sharp^\kappa_\gamma\phi\in A$, then
    $\kappa<\omega$;
  \item\label{item:sharp2} $\sharp^\kappa_\gamma\phi\in A$ iff
    $\gamma(\phi,\sharp^{\kappa-1}_{\gamma}\phi)\in A$;
  \item\label{item:flat} $\flat_\gamma\phi\in A$ iff
    $\gamma(\phi,\flat_{\gamma}\phi)\in A$.
  \end{myenumerate}
\end{lemma}
% \begin{proof}[Proof sketch]
%   Using maximality and Lemma~\ref{lem:linear}.
% \end{proof}
\begin{proof}
  \emph{\ref{item:and}:} If $\phi\land\psi\in A$, then $(A')^s$ is
  consistent for $A':=A\cup\{\phi\}$. By maximality of $A$, it follows
  that $t$ is not injective on $A'$, i.e.\ there is $\phi'\in A$ such
  that $(\phi')^t=\phi^t$. Again by maximality,
  $\phi\npreceq\phi'$, so $\phi'\preceq\phi$ by
  Lemma~\ref{lem:linear}, and hence $\phi\in\upA$.

  \emph{\ref{item:or}:} If $\phi\lor\psi\in A$, then either
  $(A\cup\{\phi\})^s$ or $(A\cup\{\psi\})^s$ is consistent; in both
  cases, proceed as for \ref{item:and}.

  \emph{\ref{item:bot}:} Clear.

  \emph{\ref{item:sharp1}:} If $\sharp^\kappa_\gamma\phi\in A$, then
  by Lemma~\ref{lem:constructive-dual} and by finiteness of $A$, there
  is $i<\omega$ such that $(A\cup\{\sharp^i_\gamma\phi\})^s$ is
  consistent. By maximality of $A$, $\kappa\le i$.

  \emph{\ref{item:sharp2}:} Both formulas have $\sharp_\gamma\phi$ as
  the greatest $\sharp$-subformula of their $t$-translation, and their
  $s$-translations are syntactically equal. Therefore if, e.g.,
  $\sharp^\kappa_\gamma\phi\in A$, then $(A')^s$, where
  $A'=A\cup\{\gamma(\phi,\sharp^{\kappa-1}_{\gamma}\phi)\}$, is
  consistent; moreover $(A')^t\subseteq\Sigma$. By maximality of $A$,
  it follows that the translation $t$ does not remain injective on
  $A'$, i.e.\ we have
  $\chi^t=(\gamma(\phi,\sharp^{\kappa-1}_{\gamma}\phi))^t$ for some
  $\chi\in A$. Again by maximality of $A$ and Lemma~\ref{lem:linear},
  we must have
  $\chi\preceq\gamma(\phi,\sharp^{\kappa-1}_{\gamma}\phi)$, so that
  $\gamma(\phi,\sharp^{\kappa-1}_{\gamma}\phi)\in\upA$. Similarly, we
  show that $\gamma(\phi,\sharp^{\kappa-1}_{\gamma}\phi)\in A$ implies
  $\sharp^\kappa_\gamma\phi\in \upA$. Together, these implications
  yield the same implications for $A$ in place of~$\upA$: e.g.\ if
  $\sharp^i_\gamma\phi\in A$ for $i<\omega$ (where necessarily $i>0$),
  then $\gamma(\phi,\sharp^{i-1}_{\gamma}\phi)\in\upA$, so
  $\gamma(\phi,\sharp^{j}_{\gamma}\phi)\in A$ for some $j\le i-1$.
  Thus $\sharp^{j+1}_{\gamma}\phi\in\upA$, so that $i\le j+1$,
  implying $j=i-1$.

  \emph{\ref{item:flat}:} Similar to (and easier
  than)~\ref{item:sharp2}.\rightqed
\end{proof}

\noindent We denote by $\CT_\Sigma$ the (by
Lemma~\ref{lem:to-at-fin}, finite) set of timed-out $\Sigma$-atoms,
and proceed to construct a model with carrier $\CT_\Sigma$. We need to
construct a coalgebra rather than just a relational structure; this
coalgebra should adequately implement the formulas contained in the
states. This requirement is encapsulated in the notion of
\emph{coherence}:

\begin{definition}\label{def:coherence}
  A coalgebra structure $\xi$ on $\CT_\Sigma$ is \emph{coherent} if
  for every $A\in\CT_\Sigma$ and every timed-out $\Sigma$-formula
  $\spadesuit\rho$ where $\spadesuit$ is either a modal operator
  $\hearts\in\Lambda$ or a dual modal operator $\dual\hearts$,
  \begin{equation*}
    \xi(A)\in\Sem{\spadesuit}_{\CT_\Sigma}(\hat\rho)\text{ whenever }\spadesuit\rho\in A,
  \end{equation*}
  where $\hat\rho=\{B\in \CT_\Sigma\mid \rho\in B\}$.
  % (Note that `if' is
  % actually sufficient, as the negation of $\hearts\phi$ is
  % $\bar\hearts\neg\phi$.)
 \end{definition}
% (This definition is intentionally more general than needed
% for the completeness proof, as we will reuse it in the next section
% when we analyse complexity.)
 \noindent Existence of such a coherent coalgebra structure relies on
 one-step completeness.

 \begin{lemma}[Existence lemma]\label{lem:existence}
  There exists a coherent coalgebra structure on~$\CT_\Sigma$.
\end{lemma}
\begin{proof}
  The proof follows the same pattern as the one for the fixpoint-free
  case~\cite{Schroder07}. We can construct the coalgebra structure
  $\xi$ pointwise. So let $A\in\CT_\Sigma$; in the notation of
  Definition~\ref{def:coherence}, we have to show that there exists
  $t\in T\CT_\Sigma$ such that
  $t\in\Sem{\spadesuit}_{\CT_\Sigma}(\hat\phi)$ whenever
  $\spadesuit\phi\in A$. With a view to deriving a contradiction,
  assume the contrary. Then $T\CT_\Sigma,\tau\models\psi$ where
  \begin{align*}
    \psi&=\textstyle\Lor_{\hearts\rho\in A}\neg\hearts a_{\hearts\rho}
    \lor\Lor_{\dual\hearts\rho\in A}\hearts b_{\dual\hearts\rho}\\
    \tau(a_{\hearts\rho})&=\hat\rho\\
    \tau(b_{\dual\hearts\rho})&=\CT_\Sigma-\hat\rho
  \end{align*}
  for pairwise distinct propositional variables
  $a_{\hearts\rho},b_{\dual\hearts\rho}$. By one-step completeness, it
  follows that $\psi$ is provable over $\CT_\Sigma,\tau$, so by
  Lemma~\ref{lem:mon-rules-complete} there is an $\Rules$-derivable
  monotone one-step rule $\phi/\chi\in\Rules$ such that
  $\chi\PLentails\psi$ and $\CT_\Sigma,\tau\models\phi$. Now let
  $\theta$ be the substitution defined by
  $\theta(a_{\hearts\rho})=\rho^s$ and
  $\theta(b_{\dual\hearts\rho})=\neg(\rho^s)$.  Then $A^s$
  propositionally entails $\neg\psi\theta$.  Since~$A^s$ is
  consistent, we are done once we show that $\psi\theta$ is
  provable. To this end, it suffices to show that $\phi\theta$ is
  provable. Assume the contrary, i.e.\ $\neg\phi\theta$ is
  consistent. Then there is a conjunctive clause $\phi_0$ in the
  disjunctive normal form of $\neg\phi$ such that $\phi_0\theta$ is
  consistent. Since $\phi/\chi$ is monotone and $\chi\PLentails\psi$
  just means that $\chi$ is contained in $\psi$, $\phi$ is negative in
  the $a_{\hearts\rho}$ and positive in the $b_{\dual\hearts\rho}$;
  thus, $\phi_0$ is positive in the $a_{\hearts\rho}$ and negative in
  the $b_{\dual\hearts\rho}$. Therefore, $\phi_0\theta$ is, after
  removing double negations, of the form $\Land A^s$ for a set $A$ of
  timed-out $\Sigma$-formulas. By the timed-out Lindenbaum lemma
  (Lemma~\ref{lem:to-at-fin}), $A$ is contained in a timed-out
  $\Sigma$-atom $B\in\CT_\Sigma$. Then $B\in\Sem{\phi_0}\tau$, so
  $B\notin\Sem{\phi}\tau$, in contradiction to
  $\CT_\Sigma,\tau\models\phi$. \rightqed
\end{proof} 
\noindent Next, we establish that a coherent coalgebra does what we
expect:
\begin{lemma}[Truth lemma]\label{lem:truth}
  Let $\xi$ be a coherent coalgebra structure on $\CT_\Sigma$. If
  $A\in\CT_\Sigma$ and $\phi\in A$, then
  $A\models_{(\CT_\Sigma,\xi)} \phi^s$.
\end{lemma}
 \begin{proof}
   Induction over timed-out $\Sigma$-formulas $\phi$ using the
   lexicographic product of the subterm ordering on $\phi^t$ with
   $\preceq$ as the induction measure. We note that by
   Lemma~\ref{lem:trans-entail} (and soundness), the inductive
   hypothesis can be strengthened to apply also to $\phi\in\upA$. The
   case $\phi=\top$ is trivial. The steps for $\bot$, $\land$, and
   $\lor$ are taken care of by Lemma~\ref{lem:pred-at}. The step for
   modal operators $\hearts$ or $\dual\hearts$ is by coherence and
   monotonicity.

   Next, we discharge the case $\phi=\sharp^\kappa_\gamma\psi$. By
   Lemma~\ref{lem:pred-at}, we have $\kappa<\omega$ and
   $\gamma(\psi,\sharp^{\kappa-1}_\gamma\psi)\in A$. We prove by a
   further induction on modal fixpoint schemes $\delta$ that
   $\delta(\psi,\sharp^{\kappa-1}_\gamma\psi)\in\upA$ implies that
   $A\models_{(\CT_\Sigma,\xi)}(\delta(\psi,\sharp^{\kappa-1}_\gamma\psi))^s$. The
   case for the parameter variable is discharged by the inductive
   hypothesis applied to $\psi$, as $\psi^t$ is a proper subterm of
   $\phi^t$, while the case for the recursion variable is discharged
   by the inductive hypothesis applied to
   $\sharp^{\kappa-1}_\gamma\psi$ (which is strictly below
   $\sharp^\kappa_\gamma\psi$ w.r.t.~$\preceq$ since
   $\kappa<\omega$). The cases for Boolean operations and modal
   operators are as in the outer induction. This finishes the inner
   induction, so that
   $A\models_{(\CT_\Sigma,\xi)}\gamma(\psi,\sharp^{\kappa-1}_\gamma\psi)$
   and hence $A\models_{(\CT_\Sigma,\xi)}\sharp^\kappa_\gamma\psi$.
  
   \newcommand{\TExt}[1]{\CT_\Sigma(#1)}

   Finally, the case $\phi=\flat_\gamma\psi$ is discharged by
   coinduction. For timed-out $\Sigma$-formulas $\rho$, we put
   $\CT_\Sigma(\rho)=\{A\in \CT_\Sigma\mid \rho\in A\}$ (this is like
   $\hat\rho$ in Definition~\ref{def:coherence} but will be applied to
   long expressions).  As $\Sem{\flat_\gamma\psi}_{(\CT_\Sigma,\xi)}$
   is a greatest fixpoint, it suffices to prove that
   $\TExt{\flat_\gamma\psi}$ is semantically a postfixpoint of
   $\gamma(\psi)$, i.e.\
   \begin{equation}\label{eq:goal}
     \TExt{\flat_\gamma\psi}\subseteq
     \Sem{\gamma(\psi)}(\TExt{\flat_\gamma\psi}).
   \end{equation}
   To begin, we prove by induction on modal fixpoint schemes $\delta$
   that for all timed-out $\Sigma$-formulas $\chi$,
   \begin{equation}\label{eq:hat-eval}
     \TExt{\delta(\psi,\chi)}\subseteq\Sem{\delta(\psi)}(\CT_\Sigma(\chi)).
   \end{equation}
   The cases for $\bot$, $\top$, and the recursion variable $x$ are
   clear. The case for the parameter variable $p$ is discharged by the
   outer inductive hypothesis applied to $\psi$. The cases for $\land$
   and $\lor$ are by Lemma~\ref{lem:pred-at}; e.g., we have
  \begin{align*}
    & \TExt{(\delta_1\lor\delta_2)(\psi,\chi)} \\
    & \subseteq \TExt{\delta_1(\psi,\chi)}\cup\TExt{\delta_2(\psi,\chi)}
    && \text{(Lemma~\ref{lem:pred-at})}\\
    &\subseteq\Sem{\delta_2(\psi)}(\CT_\Sigma(\chi))\cup
      \Sem{\delta_1(\psi)}(\CT_\Sigma(\chi))
    && \text{(induction)}\\
    & = \Sem{(\delta_2\lor\delta_1)(\psi)}(\CT_\Sigma(\chi)).
  \end{align*}
  Finally, the case for modal operators is by coherence: for
  $\spadesuit$ being either a modal operator $\hearts\in\Lambda$ or a
  dual modal operator $\dual\hearts$, we have
  \begin{align*}
    & \TExt{\spadesuit\delta(\psi,\chi)} \\
    & \subseteq \xi^{-1}[\Sem{\spadesuit}_{\CT_\Sigma}(\TExt{\delta(\psi,\chi)})]
    && \text{(coherence)}\\
    &\subseteq \xi^{-1}[\Sem{\spadesuit}_{\CT_\Sigma}(\Sem{\delta(\psi)}{\CT_\Sigma(\chi)})]
    && \text{(induction, monotonicity)}\\
    &= \Sem{\spadesuit\delta(\psi)}(\CT_\Sigma(\chi)).
  \end{align*}
  By (\ref{eq:hat-eval}), we reduce our goal (\ref{eq:goal}) to
  \begin{equation*}
    \TExt{\flat_\gamma\phi}\subseteq\TExt{\gamma(\psi,\flat_\gamma\psi)},
  \end{equation*}
  which holds by Lemma~\ref{lem:pred-at}. \rightqed
 \end{proof}
 % \begin{proof}[sketch]
 %   Induction over timed-out $\Sigma$-formulas $\phi$ using the
 %   lexicographic product of the subterm ordering on $\phi^t$ and
 %   $\preceq$ as the induction measure, and with the inductive
 %   hypothesis strengthened to apply also to
 %   $\phi\in\upA$. Boolean cases are by
 %   Lemma~\ref{lem:pred-at}; the step for modal operators is by
 %   coherence.  The case for $\flat$-operators is by coinduction. For
 %   $\phi=\sharp_\gamma(\psi)^\kappa$, we have $\kappa<\omega$ and
 %   $\gamma(\psi,\sharp_\gamma(\psi)^{\kappa-1})\in\upA$ by
 %   Lemma~\ref{lem:pred-at}. Then prove by a further induction over
 %   subformulas $\delta$ of $\gamma$ that
 %   $A\models_{(\CT_\Sigma,\xi)}(\delta(\psi,\sharp_\gamma(\psi)^{\kappa-1}))^s$
 %   whenever
 %   $\delta(\psi,\sharp_\gamma(\psi)^{\kappa-1})\in\upA$. Here,
 %   the case for the parameter variable $x$ is discharged by the
 %   inductive hypothesis applied to $\sharp_\gamma(\psi)^{\kappa-1}$. \qed
 % \end{proof}

\noindent In summary, we have proved completeness of the Kozen-Park
axiomatization:
 
\begin{theorem}[Completeness]\label{thm:main}
  If $\Gamma$ is admissible and $\Rules$ is one-step complete, then
  the logic $\Lang_\sharp$ is complete over finite models.
\end{theorem}
\begin{proof}
  We have to show that every consistent formula $\phi$ is satisfiable
  in a finite $T$-coalgebra, where by Lemma~\ref{lem:NNF} we can
  assume that $\phi$ is in NNF. This follows from the preceding lemmas
  by the usual pattern: Let $\Sigma$ be the least Fischer-Ladner
  closed set of formulas in NNF containing $\phi$. Then~$\phi$ is
  contained in the $t$-image of some timed-out $\Sigma$-atom by the
  timed-out Lindenbaum lemma (Lemma~\ref{lem:to-at-fin}), hence by the
  truth lemma (Lemma~\ref{lem:truth}) satisfied in a coherent
  coalgebra on the set $\CT_\Sigma$ of timed-out $\Sigma$-atoms, which
  exists by Lemma~\ref{lem:existence}; the set $\CT_\Sigma$ is finite
  by Lemma~\ref{lem:to-at-fin}. \rightqed
\end{proof}
\noindent 
We enumerate a few concrete instances of this result:

\begin{example}
  We obtain that the Kozen-Park axiomatization (in combination with
  the modal rules and propositional reasoning) is complete for the
  following logics.
  \begin{enumerate}
  \item All admissible flat fragments of the standard relational
    $\mu$-calculus, interpreted over unrestricted, serial, or
    deterministic Kripke models. We thus recover the known
    completeness results for LTL~\cite{GabbayEA80} and serial and
    non-serial CTL~\cite{EmersonHalpern85}, but also for logics
    featuring operators outside even CTL$^*$ as discussed in
    Example~\ref{expl:logics}.\ref{item:K}. Non-serial CTL is already
    covered by the generic results of Santocanale and
    Venema~\citeyear{SantocanaleVenema10}.
  \item All admissible flat fragments of the graded $\mu$-calculus,
    including ones featuring the operators `the current state is the
    root of a finite binary tree all whose leaves satisfy \dots',
    `\dots holds somewhere on every infinite $k+1$-ary tree starting
    at the current state', and `the current state is the root of a
    finite binary tree all whose leaves are at even distance from the
    root and satisfy \dots' discussed in
    Example~\ref{expl:logics}.\ref{item:graded}.
  \item All admissible flat probabilistic fixpoint logics, including
    ones featuring linear inequalities on probabilities.
  \item All admissible flat conditional fixpoint logics.
  \item All admissible flat fragments of the monotone $\mu$-calculus
    and the serial monotone $\mu$-calculus, including the
    star-nesting-free fragments of CPDL and game logic. 
  \item All admissible flat fragments of the alternating-time
    $\mu$-calculus, including alternating-time temporal logic ATL but
    also logics going beyond ATL, e.g.\ ones featuring the operator
    `\dots holds in all even states along any path' discussed in
    Example~\ref{expl:logics}.\ref{item:amc}. In fact, ATL appears to
    be the only example outside the relational world for which a
    completeness result of this type was previously
    known~\cite{GorankoVanDrimmelen06}.
  \end{enumerate}
  In most of these examples except graded and probabilistic fixpoint
  logics, we actually obtain finite axiomatizability, and locally
  finite axiomatizability in the case of graded fixpoint logics
  (Remark~\ref{rem:finite-ax}).
\end{example}

\section{Conclusions}\label{sec:conclusions}

\noindent We have lifted the completeness theorem for flat modal
fixpoint logics~\cite{SantocanaleVenema10} to the level of generality
of coalgebraic logic. Specifically, we have given a Kozen-Park style
axiomatization for fixpoint operators, and we have shown this
axiomatization to be sound and complete under the conditions that (i)
the defining formulas of the fixpoint operators satisfy a mild
syntactic restriction, and (ii) the coalgebraic base logic is
axiomatized by a one-step complete rule set.  This result covers,
e.g., probabilistic fixpoint logics and flat fragments of the monotone
$\mu$-calculus, the ambient fixpoint logic of Parikh's game
logic~\cite{Parikh85} and concurrent PDL~\cite{Peleg87}. Further
instances include completeness of flat fragments of the graded
$\mu$-calculus~\cite{KupfermanEA02}, to our knowledge the first
completeness result for any graded fixpoint logic, and completeness of
flat fragments of the alternating-time $\mu$-calculus~\cite{AlurEA02},
with alternating-time temporal logic (ATL) being apparently the only
previously known example~\cite{GorankoVanDrimmelen06}. In those
examples that have finite modal similarity type, in particular for
alternating-time fixpoint logics, the axiomatization we obtain is
finite.

A core technical point in the proof was to show that essentially all
monotone modal operators (including nested ones like $\Box\Box$, as
long as the nesting depth is uniform) are finitary $\CO$-adjoints in
the sense of Santocanale~\citeyear{Santocanale08}, and hence induce
\emph{constructive} fixpoint operators that can be approximated in
$\omega$ steps in the Lindenbaum algebra. This has enabled a model
construction using explicit time-outs for least fixpoint formulas in
the spirit of the completeness proof for the aconjunctive fragment of
the $\mu$-calculus~\cite{Kozen83}, which relies on a judicious
definition of timed-out formula.

A remaining open problem is to extend the completeness result to
larger fragments of the coalgebraic $\mu$-calculus beyond the single
variable fragment covered here, first and foremost the
alternation-free fragment, and eventually the full coalgebraic
$\mu$-calculus.

\bibliographystyle{ACM-Reference-Format-Journals}
  \bibliography{coalgml}


\providecommand{\noopsort}[1]{}
\begin{thebibliography}{00}

%%% the Web address.

\ifx \showCODEN    \undefined \def \showCODEN     #1{\unskip}     \fi
\ifx \showDOI      \undefined \def \showDOI       #1{{\tt DOI:}\penalty0{#1}\ }
  \fi
\ifx \showISBNx    \undefined \def \showISBNx     #1{\unskip}     \fi
\ifx \showISBNxiii \undefined \def \showISBNxiii  #1{\unskip}     \fi
\ifx \showISSN     \undefined \def \showISSN      #1{\unskip}     \fi
\ifx \showLCCN     \undefined \def \showLCCN      #1{\unskip}     \fi
\ifx \shownote     \undefined \def \shownote      #1{#1}          \fi
\ifx \showarticletitle \undefined \def \showarticletitle #1{#1}   \fi
\ifx \showURL      \undefined \def \showURL       #1{#1}          \fi

\bibitem[\protect\citeauthoryear{Alur, Henzinger, and Kupferman}{Alur
  et~al\mbox{.}}{2002}]%
        {AlurEA02}
R.~Alur, T.~Henzinger, and O.~Kupferman. 2002.
\newblock \showarticletitle{Alternating-time temporal logic}.
\newblock {\em J.\ ACM\/}  {49} (2002), 672--713.
\newblock


\bibitem[\protect\citeauthoryear{Baader, Calvanese, McGuinness, Nardi, and
  Patel-Schneider}{Baader et~al\mbox{.}}{2003}]%
        {BaaderEA03}
F.~Baader, D.~Calvanese, D.~McGuinness, D.~Nardi, and P.~Patel-Schneider
  (Eds.). 2003.
\newblock {\em The Description Logic Handbook}.
\newblock Cambridge University Press.
\newblock
\showISBNx{0-521-78176-0}


\bibitem[\protect\citeauthoryear{Chellas}{Chellas}{1980}]%
        {Chellas80}
B.~Chellas. 1980.
\newblock {\em Modal Logic}.
\newblock Cambridge University Press.
\newblock


\bibitem[\protect\citeauthoryear{C{\^{\i}}rstea, Kupke, and
  Pattinson}{C{\^{\i}}rstea et~al\mbox{.}}{2011}]%
        {CirsteaEA11}
C.~C{\^{\i}}rstea, C.~Kupke, and D.~Pattinson. 2011.
\newblock \showarticletitle{{EXPTIME} Tableaux for the Coalgebraic
  {$\mu$}-Calculus}.
\newblock {\em Log.\ Meth.\ Comput.\ Sci.\/} {7}, 3.03 (2011), 1--33.
\newblock


\bibitem[\protect\citeauthoryear{C{\^{\i}}rstea and Pattinson}{C{\^{\i}}rstea
  and Pattinson}{2007}]%
        {CirsteaPattinson07}
C.~C{\^{\i}}rstea and D.~Pattinson. 2007.
\newblock \showarticletitle{Modular construction of complete coalgebraic
  logics}.
\newblock {\em Theoret.\ Comput.\ Sci.\/}  {388} (2007), 83--108.
\newblock


\bibitem[\protect\citeauthoryear{D'Agostino and Visser}{D'Agostino and
  Visser}{2002}]%
        {DAgostinoVisser02}
G.~D'Agostino and A.~Visser. 2002.
\newblock \showarticletitle{Finality regained: A coalgebraic study of
  {Scott}-sets and multisets}.
\newblock {\em Arch.\ Math.\ Logic\/}  {41} (2002), 267--298.
\newblock


\bibitem[\protect\citeauthoryear{Dam}{Dam}{1994}]%
        {Dam94}
M.~Dam. 1994.
\newblock \showarticletitle{{CTL}* and {ECTL}* as Fragments of the Modal
  mu-Calculus}.
\newblock {\em Theoret.\ Comput.\ Sci.\/}  {126} (1994), 77--96.
\newblock


\bibitem[\protect\citeauthoryear{{De Caro}}{{De Caro}}{1988}]%
        {DeCaro88}
F.~{De Caro}. 1988.
\newblock \showarticletitle{Graded modalities {II}}.
\newblock {\em Stud.\ Log.\/}  {47} (1988), 1--10.
\newblock


\bibitem[\protect\citeauthoryear{Emerson}{Emerson}{1990}]%
        {Emerson90}
E.~Emerson. 1990.
\newblock \showarticletitle{Temporal and Modal Logic}.
\newblock In {\em Handbook of Theoretical Computer Science, Volume {B:} Formal
  Models and Semantics}. Elsevier, 995--1072.
\newblock


\bibitem[\protect\citeauthoryear{Emerson and Clarke}{Emerson and
  Clarke}{1982}]%
        {EmersonClarke82}
E.~Emerson and E.~Clarke. 1982.
\newblock \showarticletitle{Using Branching Time Temporal Logic to Synthesize
  Synchronization Skeletons}.
\newblock {\em Sci.\ Comput.\ Program.\/}  {2} (1982), 241--266.
\newblock


\bibitem[\protect\citeauthoryear{Emerson and Halpern}{Emerson and
  Halpern}{1985}]%
        {EmersonHalpern85}
E.~Emerson and J.~Halpern. 1985.
\newblock \showarticletitle{Decision Procedures and Expressiveness in the
  Temporal Logic of Branching Time}.
\newblock {\em J.\ Comput.\ System Sci.\/}  {30} (1985), 1--24.
\newblock


\bibitem[\protect\citeauthoryear{Emerson and Lei}{Emerson and Lei}{1986}]%
        {EmersonLei86}
E.~Emerson and C.-L. Lei. 1986.
\newblock \showarticletitle{Efficient Model Checking in Fragments of the
  Propositional Mu-Calculus}. In {\em Logic in Computer Science, LICS 86}.
  IEEE, 267--278.
\newblock


\bibitem[\protect\citeauthoryear{Fagin and Halpern}{Fagin and Halpern}{1994}]%
        {FaginHalpern94}
R.~Fagin and J.~Halpern. 1994.
\newblock \showarticletitle{Reasoning about knowledge and probability}.
\newblock {\em J.\ ACM\/}  {41} (1994), 340--367.
\newblock


\bibitem[\protect\citeauthoryear{Fine}{Fine}{1972}]%
        {Fine72}
K.~Fine. 1972.
\newblock \showarticletitle{In so many possible worlds}.
\newblock {\em Notre Dame J.\ Formal Logic\/}  {13} (1972), 516--520.
\newblock


\bibitem[\protect\citeauthoryear{Friedman and Halpern}{Friedman and
  Halpern}{1994}]%
        {FriedmanHalpern94}
N.~Friedman and J.~Halpern. 1994.
\newblock \showarticletitle{On the Complexity of Conditional Logics}.
\newblock In {\em Knowledge Representation and Reasoning, KR 94}, {Jon Doyle},
  {Erik Sandewall}, {and} {Pietro Torasso} (Eds.). Morgan Kaufmann, 202--213.
\newblock


\bibitem[\protect\citeauthoryear{Gabbay, Pnueli, Shelah, and Stavi}{Gabbay
  et~al\mbox{.}}{1980}]%
        {GabbayEA80}
D.~Gabbay, A.~Pnueli, S.~Shelah, and J.~Stavi. 1980.
\newblock \showarticletitle{On the Temporal Basis of Fairness}. In {\em
  Principles of Programming Languages, POPL 1980}, {Paul Abrahams}, {Richard
  Lipton}, {and} {Stephen Bourne} (Eds.). {ACM} Press, 163--173.
\newblock
\showISBNx{0-89791-011-7}


\bibitem[\protect\citeauthoryear{Goranko and van Drimmelen}{Goranko and van
  Drimmelen}{2006}]%
        {GorankoVanDrimmelen06}
V.~Goranko and G.~van Drimmelen. 2006.
\newblock \showarticletitle{Complete axiomatization and decidability of
  Alternating-time temporal logic}.
\newblock {\em Theoret.\ Comput.\ Sci.\/}  {353} (2006), 93--117.
\newblock


\bibitem[\protect\citeauthoryear{Hausmann and Schr\"oder}{Hausmann and
  Schr\"oder}{2015}]%
        {HausmannSchroder15}
D.~Hausmann and L.~Schr\"oder. 2015.
\newblock \showarticletitle{Global Caching for the Flat Coalgebraic
  {$\mu$}-Calculus}. In {\em Temporal Representation and Reasoning, TIME 2015},
  {Fabio Grandi}, {Martin Lange}, {and} {Alessio Lomuscio} (Eds.). IEEE,
  121--130.
\newblock


\bibitem[\protect\citeauthoryear{Heifetz and Mongin}{Heifetz and
  Mongin}{2001}]%
        {HeifetzMongin01}
A.~Heifetz and P.~Mongin. 2001.
\newblock \showarticletitle{Probabilistic logic for type spaces}.
\newblock {\em Games and Economic Behavior\/}  {35} (2001), 31--53.
\newblock


\bibitem[\protect\citeauthoryear{Huth and Kwiatkowska}{Huth and
  Kwiatkowska}{1997}]%
        {HuthKwiatkowska97}
M.~Huth and M.~Kwiatkowska. 1997.
\newblock \showarticletitle{Quantitative Analysis and Model Checking}. In {\em
  Logic in Computer Science, LICS 1997}. {IEEE} Computer Society, 111--122.
\newblock
\showISBNx{0-8186-7925-5}


\bibitem[\protect\citeauthoryear{Kozen}{Kozen}{1983}]%
        {Kozen83}
D.~Kozen. 1983.
\newblock \showarticletitle{Results on the propositional {$\mu$-calculus}}.
\newblock {\em Theoret.\ Comput.\ Sci.\/}  {27} (1983), 333--354.
\newblock


\bibitem[\protect\citeauthoryear{Kozen and Parikh}{Kozen and Parikh}{1981}]%
        {KozenParikh81}
D.~Kozen and R.~Parikh. 1981.
\newblock \showarticletitle{An elementary proof of the completeness of {PDL}}.
\newblock {\em Theoret.\ Comput.\ Sci.\/}  {14} (1981), 113--118.
\newblock


\bibitem[\protect\citeauthoryear{Kupferman, Sattler, and Vardi}{Kupferman
  et~al\mbox{.}}{2002}]%
        {KupfermanEA02}
O.~Kupferman, U.~Sattler, and M.~Vardi. 2002.
\newblock \showarticletitle{The Complexity of the Graded {$\mu$}-Calculus}. In
  {\em Automated Deduction, CADE 2002}{\em , LNCS}, {Andrei Voronkov} (Ed.),
  Vol. 2392. Springer, 423--437.
\newblock
\showISBNx{3-540-43931-5}


\bibitem[\protect\citeauthoryear{Kupke and Pattinson}{Kupke and
  Pattinson}{2010}]%
        {KupkePattinson10}
C.~Kupke and D.~Pattinson. 2010.
\newblock \showarticletitle{On Modal Logics of Linear Inequalities}. In {\em
  Advances in Modal Logic, AiML 2010}, {Lev Beklemishev}, {Valentin Goranko},
  {and} {Valentin Shehtman} (Eds.). College Publications, 235--255.
\newblock
\showISBNx{978-1-84890-013-4}


\bibitem[\protect\citeauthoryear{Larsen and Skou}{Larsen and Skou}{1991}]%
        {LarsenSkou91}
K.~Larsen and A.~Skou. 1991.
\newblock \showarticletitle{Bisimulation through probabilistic testing}.
\newblock {\em Inf.\ Comput.\/}  {94} (1991), 1--28.
\newblock
\showISSN{0890-5401}


\bibitem[\protect\citeauthoryear{Laver}{Laver}{1976}]%
        {lave:well76}
R.~Laver. 1976.
\newblock \showarticletitle{Well-quasi-orders and sets of finite sequences}.
\newblock {\em Math.\ Proc.\ Cambridge Philos.\ Soc.\/}  {79} (1976), 1--10.
\newblock


\bibitem[\protect\citeauthoryear{Lewis}{Lewis}{1969}]%
        {Lewis69}
D.~Lewis. 1969.
\newblock {\em Convention, {A} Philosophical Study}.
\newblock Harvard University Press.
\newblock


\bibitem[\protect\citeauthoryear{Lichtenstein and Pnueli}{Lichtenstein and
  Pnueli}{2000}]%
        {LichtensteinPnueli00}
O.~Lichtenstein and A.~Pnueli. 2000.
\newblock \showarticletitle{Propositional Temporal Logics: Decidability and
  Completeness}.
\newblock {\em Logic Journal of the {IGPL}\/}  {8} (2000), 55--85.
\newblock


\bibitem[\protect\citeauthoryear{Liu, Song, Wang, and Zhang}{Liu
  et~al\mbox{.}}{2015}]%
        {LiuEA15}
W.~Liu, L.~Song, J.~Wang, and L.~Zhang. 2015.
\newblock \showarticletitle{A Simple Probabilistic Extension of Modal
  Mu-calculus}. In {\em International Joint Conference on Artificial
  Intelligence, {IJCAI} 2015}, {Qiang Yang} {and} {Michael Wooldridge} (Eds.).
  {AAAI} Press, 882--888.
\newblock
\showISBNx{978-1-57735-738-4}


\bibitem[\protect\citeauthoryear{Morgan and McIver}{Morgan and McIver}{1997}]%
        {MorganMcIver97}
C.~Morgan and A.~McIver. 1997.
\newblock \showarticletitle{A Probabilistic Temporal Calculus Based on
  Expectations}. In {\em Formal Methods Pacific 1997}. Springer, 4--22.
\newblock


\bibitem[\protect\citeauthoryear{Myers, Pattinson, and Schr{\"{o}}der}{Myers
  et~al\mbox{.}}{2009}]%
        {MyersEA09}
R.~Myers, D.~Pattinson, and L.~Schr{\"{o}}der. 2009.
\newblock \showarticletitle{Coalgebraic Hybrid Logic}. In {\em Foundations of
  Software Science and Computational Structures, {FOSSACS} 2009}{\em , LNCS},
  {Luca de~Alfaro} (Ed.), Vol. 5504. Springer, 137--151.
\newblock
\showISBNx{978-3-642-00595-4}


\bibitem[\protect\citeauthoryear{Olivetti, Pozzato, and Schwind}{Olivetti
  et~al\mbox{.}}{2007}]%
        {OlivettiEA07}
N.~Olivetti, G.~Pozzato, and C.~Schwind. 2007.
\newblock \showarticletitle{A sequent calculus and a theorem prover for
  standard conditional logics}.
\newblock {\em ACM Trans.\ Comput.\ Log\/} {8}, 4 (2007), 22:1--51.
\newblock


\bibitem[\protect\citeauthoryear{Parikh}{Parikh}{1985}]%
        {Parikh85}
R.~Parikh. 1985.
\newblock \showarticletitle{The logic of games and its applications}.
\newblock {\em Annals of Discrete Mathematics\/}  {24} (1985), 111--140.
\newblock


\bibitem[\protect\citeauthoryear{Pattinson}{Pattinson}{2003}]%
        {Pattinson03}
D.~Pattinson. 2003.
\newblock \showarticletitle{Coalgebraic Modal Logic: Soundness, Completeness
  and Decidability of Local Consequence}.
\newblock {\em Theoret.\ Comput.\ Sci.\/}  {309} (2003), 177--193.
\newblock


\bibitem[\protect\citeauthoryear{Pattinson}{Pattinson}{2004}]%
        {Pattinson04}
D.~Pattinson. 2004.
\newblock \showarticletitle{Expressive Logics for Coalgebras via Terminal
  Sequence Induction}.
\newblock {\em Notre Dame J.\ Formal Logic\/}  {45} (2004), 19--33.
\newblock


\bibitem[\protect\citeauthoryear{Pattinson and Schr{\"{o}}der}{Pattinson and
  Schr{\"{o}}der}{2010}]%
        {PattinsonSchroder10}
D.~Pattinson and L.~Schr{\"{o}}der. 2010.
\newblock \showarticletitle{Cut elimination in coalgebraic logics}.
\newblock {\em Inf.\ Comput.\/}  {208} (2010), 1447--1468.
\newblock


\bibitem[\protect\citeauthoryear{Pauly}{Pauly}{2002}]%
        {Pauly02}
M.~Pauly. 2002.
\newblock \showarticletitle{A Modal Logic for Coalitional Power in Games}.
\newblock {\em J.\ Log.\ Comput.\/}  {12} (2002), 149--166.
\newblock


\bibitem[\protect\citeauthoryear{Peleg}{Peleg}{1987}]%
        {Peleg87}
D.~Peleg. 1987.
\newblock \showarticletitle{Concurrent dynamic logic}.
\newblock {\em J.\ ACM\/}  {34} (1987), 450--479.
\newblock


\bibitem[\protect\citeauthoryear{Pratt}{Pratt}{1976}]%
        {Pratt76}
V.~Pratt. 1976.
\newblock \showarticletitle{Semantical Considerations on {Floyd-Hoare} Logic}.
  In {\em Foundations of Computer Science, FOCS 76}. IEEE, 109--121.
\newblock


\bibitem[\protect\citeauthoryear{Santocanale}{Santocanale}{2008}]%
        {Santocanale08}
L.~Santocanale. 2008.
\newblock \showarticletitle{Completions of {$\mu$}-algebras}.
\newblock {\em Ann.\ Pure Appl.\ Logic\/}  {154} (2008), 27--50.
\newblock


\bibitem[\protect\citeauthoryear{Santocanale and Venema}{Santocanale and
  Venema}{2010}]%
        {SantocanaleVenema10}
L.~Santocanale and Y.~Venema. 2010.
\newblock \showarticletitle{Completeness for flat modal fixpoint logics}.
\newblock {\em Ann.\ Pure Appl.\ Logic\/}  {162} (2010), 55--82.
\newblock


\bibitem[\protect\citeauthoryear{Schr{\"o}der}{Schr{\"o}der}{2007}]%
        {Schroder07}
L.~Schr{\"o}der. 2007.
\newblock \showarticletitle{A Finite Model Construction for Coalgebraic Modal
  Logic}.
\newblock {\em J.\ Log.\ Algebr.\ Prog.\/}  {73} (2007), 97--110.
\newblock


\bibitem[\protect\citeauthoryear{Schr{\"o}der and Pattinson}{Schr{\"o}der and
  Pattinson}{2007}]%
        {SchroderPattinson07}
L.~Schr{\"o}der and D.~Pattinson. 2007.
\newblock \showarticletitle{Modular algorithms for heterogeneous modal logics}.
  In {\em Automata, Languages and Programming, ICALP 07}{\em , LNCS}, {Lars
  Arge}, {Andrzej Tarlecki}, {and} {Christian Cachin} (Eds.), Vol. 4596.
  Springer, 459--471.
\newblock


\bibitem[\protect\citeauthoryear{Schr{\"o}der and Pattinson}{Schr{\"o}der and
  Pattinson}{2009}]%
        {SchroderPattinson09a}
L.~Schr{\"o}der and D.~Pattinson. 2009.
\newblock \showarticletitle{{PSPACE} Bounds for Rank-1 Modal Logics}.
\newblock {\em ACM Trans.\ Comput.\ Log.\/}  {10} (2009), 13:1--13:33.
\newblock


\bibitem[\protect\citeauthoryear{Schr{\"{o}}der, Pattinson, and
  Hausmann}{Schr{\"{o}}der et~al\mbox{.}}{2010}]%
        {SchroderEA10}
L.~Schr{\"{o}}der, D.~Pattinson, and D.~Hausmann. 2010.
\newblock \showarticletitle{Optimal Tableaux for Conditional Logics with
  Cautious Monotonicity}. In {\em European Conference on Artificial
  Intelligence, {ECAI} 2010}{\em , Frontiers Artif.\ Intell.\ Appl.}, {Helder
  Coelho}, {Rudi Studer}, {and} {Michael Wooldridge} (Eds.), Vol. 215. {IOS}
  Press, 707--712.
\newblock
\showISBNx{978-1-60750-605-8}


\bibitem[\protect\citeauthoryear{Schr{\"{o}}der and Venema}{Schr{\"{o}}der and
  Venema}{2010}]%
        {SchroderVenema10}
L.~Schr{\"{o}}der and Y.~Venema. 2010.
\newblock \showarticletitle{Flat Coalgebraic Fixed Point Logics}. In {\em
  Concurrency Theory, {CONCUR} 2010}{\em , LNCS}, {Paul Gastin} {and}
  {Fran{\c{c}}ois Laroussinie} (Eds.), Vol. 6269. Springer, 524--538.
\newblock
\showISBNx{978-3-642-15374-7}


\bibitem[\protect\citeauthoryear{Segerberg}{Segerberg}{1982}]%
        {Segerberg82}
K.~Segerberg. 1982.
\newblock \showarticletitle{A completeness theorem in the modal logic of
  programs}.
\newblock In {\em Universal Algebra and Applications}. Banach Centre
  Publications, Vol.~9. PWN -- Polish Scientific Publishers, Warsaw, 31--46.
\newblock


\bibitem[\protect\citeauthoryear{Walukiewicz}{Walukiewicz}{2000}]%
        {Walukiewicz00}
I.~Walukiewicz. 2000.
\newblock \showarticletitle{Completeness of {K}ozen's Axiomatisation of the
  Propositional {$\mu$}-Calculus}.
\newblock {\em Inf.\ Comput.\/}  {157} (2000), 142--182.
\newblock


\bibitem[\protect\citeauthoryear{Wolper}{Wolper}{1983}]%
        {Wolper83}
P.~Wolper. 1983.
\newblock \showarticletitle{Temporal Logic Can Be More Expressive}.
\newblock {\em Inf.\ Control\/}  {56} (1983), 72--99.
\newblock


\end{thebibliography}

\end{document}